\documentclass[reqno,12pt]{amsart}

\usepackage{tikz}
\usepackage{graphics,amssymb,amsmath}
\input cyracc.def
\font\tencyr=wncyr8
\def\cyr{\tencyr\cyracc}
\newcommand{\shift}{\mbox{}\hspace{4.5mm}}
\newcommand{\fix}[1]{\mathit{fix}(#1)}
\newcommand{\lin}[1]{\bar v(#1)}
\newcommand{\free}[1]{\tilde v(#1)}
\newcommand{\fixl}[1]{\mathit{fix}^-(#1)}
\newcommand{\fitl}[1]{\mathit{fit}^-(#1)}
\newcommand{\fitg}[1]{\mathit{fit}^\vee(#1)}
\newcommand{\cut}[1]{\bar f(#1)}
\newcommand{\ff}{{\bar f}}
\newcommand{\vv}{{\bar v}}
\newcommand{\plin}[1]{\stackrel{=}{v}(#1)}
\newcommand{\pcute}[1]{\stackrel{=}{e}(#1)}
\newcommand{\fit}[1]{\mathit{fit}(#1)}
\newcommand{\fixx}[1]{\mathit{fix}^X(#1)}
\newcommand{\fitx}[1]{\mathit{fit}^X(#1)}
\newcommand{\foldn}[1]{{\sc Fold}$^{#1}$}
\newcommand{\fold}[2]{{\sc Fold}$^{#1}\,(#2)$}
\newcommand{\cross}[1]{\mathrm{cr}(#1)}
\newcommand{\reals}{\mathbb{R}}
\newcommand{\function}[2]{:#1 \rightarrow #2}

\newcommand{\setdef}[2]{\left\{ \hspace{0.5mm} #1 : \hspace{0.5mm} #2 \right\}}
\newcommand{\refeq}[1]{(\ref{eq:#1})}
\newcommand{\hide}[1]{}
\newcommand{\edit}[1]{}

\newtheorem{theorem}{Theorem}[section]
\newtheorem{lemma}[theorem]{Lemma}

\newtheorem{corollary}[theorem]{Corollary}
\newtheorem{remark}[theorem]{Remark}

\newenvironment{proofof}[1]{\par\smallbreak\noindent{\it Proof~of~#1.}}%
{\unskip\nobreak\hfill \qed \par\medbreak}
\newcommand{\Case}[2]{\smallskip\par{\it Case #1:\/ #2}}
\newcommand{\Subcase}[2]{\smallskip\par{\it Subcase #1:\/ #2}}
\newenvironment{bfenumerate}%
{%
\begin{enumerate}}{\end{enumerate}}
\newcounter{oq}
\newcommand{\que}{\refstepcounter{oq}\par{\sc \theoq.}~}
\title{On collinear sets in straight line drawings}

\author{Alexander Ravsky\,$^*$}\thanks{$^*$\,%
Institute for Applied Problems of Mechanics and Mathematics,
Naukova St.\ 3-{\cyr B}, Lviv 79060, Ukraine.}

\author{Oleg Verbitsky\,$^*$$^\dag$}\thanks{$^\dag$\,%
Current address: Humboldt-Universit\"at zu Berlin,
Institut f\"ur Informatik, Unter den Linden 6,
D-10099 Berlin. Supported by the Alexander von Humboldt Foundation.}

\date{}

\begin{document}

\begin{abstract}
We consider straight line drawings of a planar graph $G$ with possible edge crossings.
The \emph{untangling problem} is to eliminate all edge crossings by moving as few
vertices as possible to new positions.
Let $\fix G$ denote the maximum number of vertices that can be left fixed
in the worst case.
In the \emph{allocation problem}, we are given
a planar graph $G$ on $n$ vertices together with an $n$-point set $X$
in the plane and have to draw $G$ without edge crossings so that
as many vertices as possible are located in $X$.
Let $\fit G$ denote the maximum number of points fitting this purpose
in the worst case.
As $\fix G\le\fit G$, we are interested in upper bounds for the latter
and lower bounds for the former parameter.

For each $\epsilon>0$, we construct an infinite sequence of graphs 
with $\fit G=O(n^{\sigma+\epsilon})$,
where $\sigma<0.99$ is a known graph-theoretic constant, namely the shortness
exponent for the class of cubic polyhedral graphs. 
To the best of our knowledge, this is the first example of graphs
with $\fit G=o(n)$.
On the other hand, we prove that $\fix G\ge\sqrt{n/30}$ for all $G$ with tree-width
at most 2. This extends the lower bound obtained by 
Goaoc et al.\ [{\it Discrete and Computational Geometry\/} 42:542--569 (2009)]
for outerplanar graphs.

Our upper bound for $\fit G$ is based on the fact that the constructed graphs
can have only few collinear vertices in any crossing-free drawing.
To prove the lower bound for $\fix G$, we show that graphs of tree-width 2
admit drawings that have large sets of collinear vertices with some additional
special properties.
\end{abstract}

\maketitle
\markleft{\sc ALEXANDER RAVSKY and OLEG VERBITSKY}

\section{Introduction}

\subsection{Basic definitions}

Let $G$ be a planar graph. The vertex set of $G$ will be denoted by $V_G$.
The letter $n$ will be reserved to always denote the number of vertices in $V_G$.
By a \emph{drawing} of $G$ we mean an arbitrary injective map
$\pi\function{V_G}{\reals^2}$.
The points in $\pi(V_G)$ will be referred to as \emph{vertices} of the drawing.
For an edge $uv$ of $G$,
the segment with endpoints $\pi(u)$ and $\pi(v)$ will be referred to as an \emph{edge}
of the drawing. Thus, we always consider \emph{straight-line} drawings.
It is quite possible that in $\pi$ we encounter edge crossings and even overlaps.
A drawing is \emph{plane} (or \emph{crossing-free}) if this does not happen.

Given a drawing $\pi$ of $G$ , define
$$
\fix{G,\pi}=\max_{\pi'\mathrm{\ plane}}|\setdef{v\in V_G}{\pi'(v)=\pi(v)}|.
$$
Given an $n$-point set $X$ in the plane, let
$$
\fixx G=\min_{\pi:\,\pi(V_G)=X}{\fix{G,\pi}}.
$$
Furthermore, we define
\begin{equation}\label{eq:deffix}
\fix G=\min_X\fixx G=\min_\pi\fix{G,\pi}.
\end{equation}
In other words, $\fix G$ is the maximum number of vertices
which can be fixed in any drawing of $G$ while \emph{untangling} it.

Given an $n$-point set $X$, consider now a related parameter
$$
\fitx G=\max_{\pi\mathrm{\ plane}}|\pi(V_G)\cap X|.
$$
In words, if we want to draw $G$ allocating its vertices at points of $X$, then
$\fitx G$ tells us how many points of $X$ can fit for this purpose.
To analyze the \emph{allocation} problem in the worst case, we define
$$
\fit G=\min_X\fitx G.
$$
Note that $\fitx G=\max_{\pi:\,\pi(V_G)=X}{\fix{G,\pi}}$.
It follows that $\fixx G\le\fitx G$ and, therefore,
$\fix G\le\fit G$.

\subsection{Known results on the untangling problem}

No efficient way for evaluating the parameter $\fix G$ is known.
Note that computing $\fix{G,\pi}$ is NP-hard \cite{merged,Ver}.
Essential efforts are needed to estimate $\fix G$ even for cycles, for which
we know bounds
$$
2^{-5/3}n^{2/3}-O(n^{1/3})\le\fix{C_n}\le O((n\log n)^{2/3})
$$
due to, respectively, Cibulka \cite{Cib} and Pach and Tardos \cite{PTa}.
In the general case Bose et al.~\cite{Bose} establish a lower bound
\begin{equation}\label{eq:bose}
\fix G\ge(n/3)^{1/4}.
\end{equation}
A better bound
\begin{equation}\label{eq:n2}
\fix G\ge\sqrt{n/2}
\end{equation} 
is proved for all trees (Bose et al.~\cite{Bose}, Goaoc et al.~\cite{merged}) 
and, more generally, outerplanar graphs (Goaoc et al.~\cite{merged}, cf.\ Corollary \ref{cor:outer} below). 

On the other hand, \cite{Bose,merged,KPRSV} provide examples of planar
graphs (even acyclic ones) with 
\begin{equation}\label{eq:fixupper}
\fix G=O(\sqrt n).
\end{equation}
In particular, for the fan graphs $F_n$ we have
\begin{equation}\label{eq:fixxFn}
\fixx{F_n}\le(2\sqrt2+o(1))\sqrt n\text{\ \ for\ every\ }X,
\end{equation}
see \cite{KPRSV}.
Cibulka \cite{Cib} establishes some general upper bounds, namely $\fix G=O(\sqrt n(\log n)^{3/2})$
for graphs whose maximum degree and diameter are bounded by a logarithmic function and
$\fix{G}=O((n\log n)^{2/3})$ for 3-connected graphs.

\subsection{Known results on the allocation problem}

The question whether or not $\fitx G=n$ has been studied in the literature,
especially for $X$ in general position.
If $X$ is in convex position, any triangulation on $X$ is outerplanar.
By this reason, $\fitx G<n$ for all non-outerplanar graphs $G$ and all
sets $X$ in convex position.
On the other hand, in \cite{GritzmannMPP91} it is proved that
$\fitx G=n$ for all outerplanar $G$ and all $X$ in general position.
Other results and references on this subject can be found, e.g., in~\cite{GarciaHHTV09}.

It is known that there are $\Omega(27.22^n)$
unlabeled planar graphs with $n$ vertices \cite{GimenezN09}, while a set of $n$ points in convex position
admits no more than $O(11.66^n)$ plain drawings \cite{FlajoletN99}. Combining the two results, we see
that, if $X$ is in convex position, then $\fitx G<n$ for almost all planar~$G$.

\subsection{Our present contribution}

We aim at proving upper bounds for $\fit G$ and lower bounds for $\fix G$.
Our approach to both problems is based on analysis of collinear sets of vertices
in straight line graph drawings. We show the relevance of the following questions.
How many collinear vertices can occur in a plane drawing of a graph $G$?
If there is a large collinear set, which useful features can it have?

Suppose that $\pi$ is a crossing-free drawing of a graph $G$.
A set of vertices $S\subseteq \pi(V_G)$ in $\pi$ is \emph{collinear} if all of them
lie on a line $\ell$. By a \emph{conformal displacement} of $S$
we mean a relocation $\delta\function S\ell$ preserving the relative order in which
the vertices in $S$ lie in $\ell$. We call $S$ \emph{free} if every
conformal displacement $\delta\function S\ell$ is extendable to a mapping
$\delta\function{\pi(V_G)}{\reals^2}$ so that $\delta\circ\pi$ is a crossing-free
drawing of $G$ (i.e., whenever we shift vertices in $S$ along $\ell$ without breaking
their relative order, then all edge crossings that may arise can be eliminated by subsequently
moving the vertices in $\pi(V_G)\setminus S$).
Let $\lin{G,\pi}$ denote the maximum size of a collinear set in $\pi$ and
$\free{G,\pi}$ the maximum size of a free collinear set in $\pi$.
Define
$$
\lin{G}=\max_{\pi\mathrm{\ plane}}\lin{G,\pi}\text{\ \ and\ \ }
\free{G}=\max_{\pi\mathrm{\ plane}}\free{G,\pi}.
$$
Obviously, $\free G\le\lin G$. These parameters have a direct relation to $\fix G$ and $\fit G$, namely
\begin{equation}\label{eq:freefixlin}
\sqrt{\free G}\le\fix G\le\fit G\le\lin G.
\end{equation}
The latter inequality follows immediately from the definitions. The first
inequality is proved as Theorem \ref{thm:fixfree} below.

In Section \ref{s:lin} we construct,
for each $\epsilon>0$, an infinite sequence of graphs with $\lin G=O(n^{\sigma+\epsilon})$
where $\sigma\le\frac{\log22}{\log23}$ is a known graph-theoretic constant, namely the shortness
exponent for the class of cubic polyhedral graphs (see Section \ref{s:prel} for the definition). 
To the best of our knowledge, this gives us the first example of graphs with
$\fit G=o(n)$. While the known upper bounds \refeq{fixupper} for $\fix G$ are still better,
note that the problems of bounding $\fix G$ and $\fit G$ from above are inequivalent.
The two parameters can be far away from one another: 
for example, in contrast with \refeq{fixxFn} we have $\fitx{F_n}\ge n-1$ for any~$X$.

By the lower bound in \refeq{freefixlin}, we have $\fix G=\Omega(\sqrt n)$ whenever
$\free G=\Omega(n)$. Therefore, identification of classes of planar graphs with linear $\free G$
is of big interest. In Section \ref{s:free} we show that
$\free G\ge n/2$ for every outerplanar graph $G$.
This gives us another proof of the bound $\fix G\ge\sqrt{n/2}$ proved for
outerplanar graphs by Goaoc et al.~\cite{merged}\footnote{%
A preliminary version of \cite{merged} gave a somewhat worse bound of $\fix G\ge\sqrt{n/3}$.
An improvement to $\sqrt{n/2}$ was made in the early version of the present paper
independently of~\cite{merged}.}.

Furthermore, we consider the broader class of graphs with tree-width at most 2.
It coincides with the class of partial 2-trees and contains also all series-parallel graphs
(see, e.g.\ \cite[Sect.\ 8.3]{Bodlaender98}). For any graph $G$ in this class, we prove
that $\free G\ge n/30$ and, therefore, $\fix G\ge\sqrt{n/30}$.
The proof of this result takes Section \ref{s:2trees}.
Note that the sublinear upper bound for $\lin G$ is established in Section \ref{s:lin}
for a sequence of graphs whose tree-width is bounded by a constant.
We conclude with the discussion of open problems in Section~\ref{s:open}.

\section{Preliminaries}\label{s:prel}

Given a planar graph $G$, we denote the number of vertices, edges,
and faces in it, respectively, by $v(G)$, $e(G)$, and $f(G)$. The latter number
does not depend on a particular plane embedding of $G$ and hence is well defined.
Moreover, for connected $G$ we have
\begin{equation}\label{eq:euler}
v(G)-e(G)+f(G)=2
\end{equation}
by Euler's formula. 

A graph is \emph{$k$-connected} if it has more than $k$ vertices and 
stays connected after removal of any $k$ vertices. 
3-connected planar graphs are called \emph{polyhedral}
as, according to Steinitz's theorem, 
these graphs are exactly the 1-skeletons of convex polyhedra.
By Whitney's theorem, all plane embeddings of a polyhedral
graph $G$ are equivalent, that is, obtainable from one another by a plane homeomorphism
up to the choice of outer face. In particular, the set of facial cycles (i.e.,
boundaries of faces) of $G$ does not depend on a particular plane embedding.

A planar graph $G$ is \emph{maximal} if adding an edge between any two 
non-adjacent vertices of $G$ violates planarity. Maximal planar graphs 
on more than 3 vertices are
3-connected. Clearly, all facial cycles in such graphs have length 3.
By this reason maximal planar graphs are also called \emph{triangulations}.
Note that for every triangulation $G$ we have
$
3f(G)=2e(G)
$.
Combined with \refeq{euler}, this gives us
\begin{equation}\label{eq:fv}
f(G)=2\,v(G)-4.
\end{equation}

The \emph{dual} of a polyhedral graph $G$ is a graph $G^*$ whose
vertices are the faces of $G$ (represented by their facial cycles).
Two faces are adjacent in $G^*$ iff they share a common edge.
$G^*$ is also a polyhedral graph. If we consider $(G^*)^*$, we obtain
a graph isomorphic to $G$. In a \emph{cubic} graph every vertex is
incident to exactly 3 edges. As easily seen, the dual of a triangulation
is a cubic graph. Conversely, the dual of any cubic polyhedral graph 
is a triangulation.

The \emph{circumference} of a graph $G$, denoted by $c(G)$, is the
length of a longest cycle in $G$. The \emph{shortness exponent}
of a class of graphs $\mathcal G$ is the limit inferior of quotients $\log c(G)/\log v(G)$
over all $G\in\mathcal G$. Let $\sigma$ denote the shortness exponent
for the class of cubic polyhedral graphs. It is known that
$$
0.753<\sigma\le\frac{\log22}{\log23}=0.985\ldots
$$
(see \cite{BilinskiJMY11} for the lower bound and \cite[Theorem 7(iv)]{GWa} for the upper bound).

\section{Graphs with small collinear sets}\label{s:lin}

We here construct a sequence of triangulations $G$ with $\lin G=o(v(G))$.
For our analysis we will need another parameter of a straight line drawing.
Given a crossing-free drawing $\pi$ of a graph $G$,
let $\cut {G,\pi}$ denote the maximum number of collinear points in the plane
such that each of them is an inner point of some face of $\pi$
and no two of them are in the same face. Let $\cut G=\max_\pi\cut{G,\pi}$.
In other words, $\cut G$ is equal to the maximum number of faces in
some straight line drawing of $G$ whose interiors can be cut by a line.
Further on, saying that a line \emph{cuts} a face, we mean that the line
intersects the interior of this face.

For the triangulations constructed below, we will show that $\lin G$
is small with respect to $v(G)$ because $\cut G$ is small with respect
to $f(G)$ (though we do not know any relation between $\lin G$ and $\cut G$ in general). 
Our construction can be thought of as a recursive procedure
for essentially decreasing the ratio $\cut G/f(G)$ at each recursion step
provided that we initially have $\cut G<f(G)$.

\begin{figure}
\centerline{\scalebox{0.85}{\includegraphics{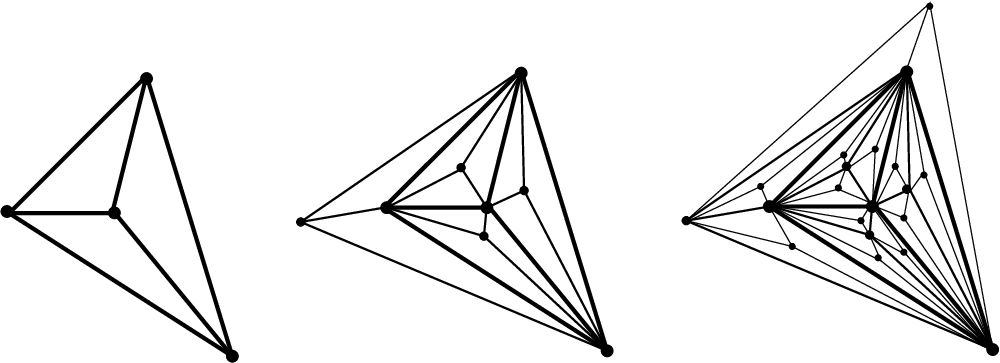}}}
\caption{An example of the construction: $G_1=K_4$, $G_2$, $G_3$.}
\label{fig:fromK4}
\end{figure}

Starting from an arbitrary triangulation $G_1$ with at least 4 vertices, 
we recursively define 
a sequence of triangulations $G_1,G_2,\ldots$. To define $G_k$, we will
describe a spherical drawing $\delta_k$ of this graph. Let $\delta_1$ be
an arbitrary drawing of $G_1$ on a sphere. Furthermore, $\delta_{i+1}$
is obtained from $\delta_i$ by triangulating each face of $\delta_i$ so
that this triangulation is isomorphic to $G_1$. An example is shown
on Fig.~\ref{fig:fromK4}. In general, upgrading $\delta_i$ to $\delta_{i+1}$
can be done in different ways, that may lead to non-isomorphic versions
of $G_{i+1}$. We make an arbitrary choice and fix the result.

\begin{lemma}\label{lem:Gk}
Denote $f=f(G_1)$, $\ff=\cut{G_1}$, and $\alpha=\displaystyle\frac{\log(\ff-1)}{\log(f-1)}$.
\begin{bfenumerate}
\item
$f(G_k)=f(f-1)^{k-1}$.
\item
$\cut{G_k}\le\ff(\ff-1)^{k-1}$.
\item
$\lin{G_k}<c\,v(G_k)^\alpha$, where $c$ is a constant depending only on~$G_1$.
\end{bfenumerate}
\end{lemma}


\begin{proof}
The first part follows from the obvious recurrence
$$
f(G_{i+1})=f(G_i)(f-1).
$$
We have to prove the other two parts.

Consider an arbitrary crossing-free straight line drawing $\pi_k$ of $G_k$.
Recall that, by construction, $G_1,\ldots,G_{k-1}$ is a chain of subgraphs of $G_k$ with
$$
V_{G_1}\subset V_{G_2}\subset\ldots\subset V_{G_{k-1}}\subset V_{G_k}.
$$
Let $\pi_i$ be the part of $\pi_k$ which is a drawing of the subgraph $G_i$.
By the Whitney theorem, $\pi_k$ can be obtained from $\delta_k$ (the spherical
drawing defining $G_k$) by an appropriate stereographic projection of the sphere
to the plane combined with a homeomorphism of the plane onto itself.
It follows that, like $\delta_{i+1}$ and $\delta_i$, drawings $\pi_{i+1}$ and $\pi_i$
have the following property: the restriction of $\pi_{i+1}$ to any face of $\pi_i$
is a drawing of $G_1$.
Given a face $F$ of $\pi_i$, the restriction of $\pi_{i+1}$ to $F$ (i.e., a plane
graph isomorphic to $G_1$) will be denoted by $\pi_{i+1}[F]$. 

Consider now an arbitrary line $\ell$. Let $\ff_i$ denote the number of faces
in $\pi_i$ cut by $\ell$. By definition, we have
\begin{equation}\label{eq:f1f}
\ff_1\le\ff.
\end{equation}
For each $1\le i<k$, we have 
\begin{equation}\label{eq:ffi1}
\ff_{i+1}\le
\left\{
\begin{array}{lcl}
\ff&\mathrm{if}&\ff_i=1,\\[1mm]
\ff_i(\ff-1)&\mathrm{if}&\ff_i>1.
\end{array}
\right.
\end{equation}
Indeed, let $K$ denote the outer face of $\pi_i$.
Equality $\ff_i=1$ means that, of all faces of $\pi_i$, $\ell$ cuts only $K$.
Within $K$, $\ell$ can cut only faces of $\pi_{i+1}[K]$ and, therefore, $\ff_{i+1}\le\ff$.

Assume that $\ff_i>1$. Within $K$, $\ell$ can now cut at most $\ff-1$ faces of $\pi_{i+1}$
(because $\ell$ cuts $\reals^2\setminus K$, a face of $\pi_{i+1}[K]$ outside $K$).
Within any inner face $F$ of $\pi_i$, $\ell$ can cut
at most $\ff-1$ faces of $\pi_{i+1}$ (the subtrahend 1 corresponds to the outer face
of $\pi_{i+1}[F]$, which surely contributes to $\ff$ but
is outside $F$). The number of inner faces $F$ cut by $\ell$ is equal to
$\ff_i-1$ (again, the subtrahend 1 corresponds to the outer face of $\pi_i$).
We therefore have 
$\ff_{i+1}\le(\ff-1)+(\ff_i-1)(\ff-1)=\ff_i(\ff-1)$,
completing the proof of~\refeq{ffi1}.

Using \refeq{f1f} and \refeq{ffi1}, a simple inductive argument gives us
\begin{equation}\label{eq:need}
\ff_i\le\ff(\ff-1)^{i-1}
\end{equation}
for each $i\le k$. As $\pi_k$ and $\ell$ are arbitrary, part 2 of the lemma is proved
by setting $i=k$ in~\refeq{need}.

To prove part 3, we have to estimate from above $\vv=|\ell\cap V(\pi_k)|$,
the number of vertices of $\pi_k$ on the line $\ell$. Put $\vv_1=|\ell\cap V(\pi_1)|$
and $\vv_i=|\ell\cap (V(\pi_i)\setminus V(\pi_{i-1}))|$ for $1<i\le k$. Clearly,
$\vv=\sum_{i=1}^k\vv_i$. Abbreviate $v=v(G_1)$. It is easy to see that
$$
\vv_1\le v-2
$$
and, for all $1<i\le k$,
$$
\vv_i\le\ff_{i-1}(v-3).
$$
It follows that
\begin{equation}\label{eq:vv}
\vv\le(v-2)+(v-3)\sum_{i=1}^{k-1}\ff_i\le\frac{(v-3)\ff}{\ff-2}(\ff-1)^{k-1},
\end{equation}
where we use \refeq{need} for the latter estimate.
It remains to express the obtained bound in terms of $v(G_k)$.
By \refeq{fv} and by part 1 of the lemma, we have $(f-1)^{k-1}<2v(G_k)/f$ and, therefore,
$$
(\ff-1)^{k-1}=(f-1)^{\alpha(k-1)}<(2/f)^\alpha\, v(G_k)^\alpha.
$$
Plugging this in to \refeq{vv}, we arrive at the desired bound for $\vv$
and hence for $\vv(G_k)$.
\end{proof}

We now need an initial triangulation $G_1$ with $\cut{G_1}<f(G_1)$. The following lemma
shows a direction where one can seek for such triangulations.

\begin{lemma}\label{lem:duals}
For every triangulation $G$ with more than 3 vertices, we have
$$\cut G\le c(G^*).$$
\end{lemma}

\begin{proof}
Given a crossing-free drawing $\pi$ of $G$ and a line $\ell$, we have to show
that $\ell$ crosses no more than $c(G^*)$ faces of $\pi$.
Shift $\ell$ a little bit to a new position $\ell'$ so that
$\ell'$ does not go through any vertex of $\pi$ and still cuts all the faces
that are cut by $\ell$.
Thus, $\ell'$ crosses boundaries of faces only via inner points of edges.
Each such crossing corresponds to transition from one vertex to another
along an edge in the dual graph $G^*$. Note that this walk is both started
and finished at the outer face of $\pi$. Since all faces are triangles,
each of them is visited at most once. Therefore, $\ell'$ determines a cycle
in $G^*$, whose length is at least the number of faces of $\pi$ cut by~$\ell$.
\end{proof}

Lemma \ref{lem:duals} suggests the following choice of $G_1$: Take
a cubic polyhedral graph $H$ approaching the infimum of the set of quotients
$\log(c(G)-1)/\log(v(G)-1)$ over all cubic polyhedral graphs $G$ and set $G_1=H^*$.
In particular, we can approach arbitrarily close to the shortness exponent
$\sigma$ defined in Section \ref{s:prel}. By Lemma \ref{lem:Gk}.3, 
we arrive at the main result of this section.

\begin{theorem}\label{thm:lin}
Let $\sigma$ denote the shortness exponent of the class of cubic polyhedral
graphs. Then for each $\alpha>\sigma$ there is a sequence of triangulations
$G$ with $\lin G=O(v(G)^\alpha)$.
\end{theorem}

\begin{corollary}
For infinitely many $n$ there is a planar graph $G$ on $n$ vertices with $\fit G=O(n^{0.99})$.
\end{corollary}

Note that the graph $G_k$ constructed in the proof of Theorem \ref{thm:lin}
is obtained from $G_{k-1}$ in a number of steps by gluing with a copy of $G_1$
at a clique of size 3. It follows that all $G_1,G_2,\ldots$ have equal tree-width.
If we take $H$ to be the Barnette-Bos\'ak-Lederberg example of a non-Hamiltonian
cubic polyhedral graph, then the construction starts with $G_1=H^*$ of tree-width
no more than~8. We, therefore, obtain the following result.

\begin{corollary}\label{cor:btw}
For some integer constant $t\le8$ and real constant $\alpha\in(0,1)$,
there is a sequence of triangulations $G$ of tree-width at most $t$
such that $\fit G\le\lin G=O(v(G)^\alpha)$.
\end{corollary}

\begin{remark}\rm
Theorem \ref{thm:lin} can be translated to a result on convex polyhedra.
Given a convex polyhedron $\pi$ and a plane $\ell$, let $\plin{\pi,\ell}=V(\pi)\cap\ell$ where
$V(\pi)$ denotes the vertex set of $\pi$. Given a polyhedral graph $G$, we define
$$
\plin G=\max_{\pi,\ell}\plin{\pi,\ell},
$$ 
where $\pi$ ranges over convex polyhedra with 1-skeleton
isomorphic to $G$. Using our construction of a sequence of triangulations $G_1,G_2,G_3,\ldots$,
we can prove that $\plin G=O(v(G)^\alpha)$ for each $\alpha>\sigma$ and infinitely many polyhedral
graphs~$G$.

Gr\"unbaum \cite{Gru} investigated the minimum number of planes which are enough to
cut all edges of a convex polyhedron $\pi$. Given a polyhedral graph $G$, define
$$
\pcute G=\max_{\pi,\ell}\pcute{\pi,\ell},
$$ 
where $\pcute{\pi,\ell}$ denotes the number of edges
that are cut by a plane $\ell$ in a convex polyhedron $\pi$ with 1-skeleton isomorphic
to $G$. Using the relation $\pcute G\le c(G^*)$, Gr\"unbaum showed 
(implicit in \cite[pp.\ 893--894]{Gru}) that $\pcute G=O(v(G)^\beta)$
for each $\beta>\log_32$ and infinitely many~$G$ (where $\log_32$ is the shortness
exponent for the class of all polyhedral graphs).
\end{remark}

\section{Graphs with large free collinear sets}\label{s:free}

Let $\pi$ be a crossing-free drawing and $\ell$ be a line.
Recall that a set $S\subset\pi(V_G)\cap\ell$ is called \emph{free} if, 
whenever we displace the vertices in $S$ along $\ell$ without violating 
their mutual order, thereby introducing edge crossings, we are able
to untangle the modified drawing by only moving the vertices in $\pi(V_G)\setminus S$. 
By $\free G$ we denote the largest size
of a free collinear set maximized over all drawings of a graph~$G$.

\begin{theorem}\label{thm:fixfree}
$\fix G\ge\sqrt{\free G}$.
\end{theorem}

\begin{proof}
Let $\fixl G$ be defined similarly to \refeq{deffix} but 
with minimization over all collinear $X$ (or over $\pi$
such that $\pi(V_G)$ is collinear). Obviously, $\fix G\le\fixl G$.
As proved in \cite{KPRSV} (based on \cite[Lemma 1]{Bose}), we actually have
\begin{equation}\label{eq:fixfixl}
\fix G=\fixl G.
\end{equation}
We use this equality here.

Suppose that $(k-1)^2<\free G\le k^2$.
By \refeq{fixfixl}, it suffices to show that any drawing 
$\pi\function{V_G}\ell$ of $G$ on a line $\ell$
can be made crossing-free with keeping $k$ vertices fixed.
Let $\rho$ be a crossing-free drawing of $G$ such that, for some
$S\subset V_G$ with $|S|>(k-1)^2$, $\rho(S)$ is a free collinear set on $\ell$.
By the Erd\H{o}s-Szekeres theorem, there exists a set $F\subset S$ of
$k$ vertices such that $\pi(F)$ and $\rho(F)$ lie on $\ell$ in the same order.
By the definition of a free set, there is a crossing-free drawing $\rho'$
of $G$ with $\rho'(F)=\pi(F)$. Thus, we can come from $\pi$ to $\rho'$ with $F$
staying fixed.
\end{proof}

Theorem \ref{thm:fixfree} sometimes gives a short way of proving
bounds of the kind $\fix G=\Omega(\sqrt n)$. For example, for the wheel graph
$W_n$ we immediately obtain $\fix{W_n}>\sqrt n-1$ from an easy observation
that $\free{W_n}=n-2$ (in fact, this repeats the argument of Pach and Tardos
for cycles \cite{PTa}). The classes of graphs with linear $\free G$ are
therefore of big interest in the context of disentanglement of drawings.
One of such classes is addressed below.

Given a drawing $\pi$, we call it \emph{track drawing} if there are parallel
lines, called \emph{tracks}, such that every vertex of $\pi$ lies on
one of the layers and every edge either lies on one of the layers or
connects endvertices lying on two consecutive layers. We call a graph
\emph{track drawable}
if it has a crossing-free track drawing. 

An obvious example of a track drawable graph is a grid graph $P_s\times P_s$.
It is also easy to see that any tree is track drawable:
two vertices are to be aligned on the same layer iff they are at the same
distance from an arbitrarily assigned root. The latter example can be considerably
extended. 

Call a drawing \emph{outerplanar} if all the vertices lie on
the outer face. An \emph{outerplanar graph} is a graph admitting
an outerplanar drawing (this definition
does not depend on whether straight line or curved drawings are considered).
The following fact is illustrated by Fig.~\ref{fig:outerlay}.

\begin{figure}
\centering
 \begin{tikzpicture}[every node/.style={circle,fill,inner sep=1pt}]
  \path[scale=.7] (0,0) node[circle,draw,color=black,fill=black!40,inner sep=1.5pt] (a)   {}
                 (0,1) node (b)   {} edge (a)
             (.5,1.5) node (c)   {}  edge (b)
              (.5,-.5) node (d) {} edge (a) 
             (1,2) node (e) {} 
                 (1,1) node (g) {} edge (b) edge (c) edge (e)
               (1,0) node (h)   {} edge (a) edge (d) edge (g)
             (1,-1) node (i)   {} edge (h) 
             (2,2) node (j)   {} edge (e)
              (2,1) node (k)   {} edge (g) edge (j)
              (2,0) node (l) {} edge (h) edge (k)
              (2,-1) node (m) {} edge (i) edge (l)
              (3,1) node (n) {} edge (k)
              (3,0) node (o) {} edge (n) edge (l)
             (4,1) node (p) {} edge (n)
              (4,0) node (q) {} edge (o) edge (p);
 \end{tikzpicture}
\qquad\qquad\qquad
 \begin{tikzpicture}[baseline=-.5cm,every node/.style={circle,fill,inner sep=1pt}]
  \path[scale=.7] (.5,-1) node[circle,draw,color=black,fill=black!40,inner sep=1.5pt] (a)   {}
                 (-2,0) node (b)   {} edge (a)
             (-1.5,1) node (c)   {}  edge (b)
              (2,0) node (d) {} edge (a) 
             (-1,1) node (e) {} 
                 (0,0) node (g) {} edge (b) edge (c) edge (e)
               (1,0) node (h)   {} edge (a) edge (d) edge (g)
             (2,1) node (i)   {} edge (h) 
             (-.5,1) node (j)   {} edge (e)
              (0,1) node (k)   {} edge (g) edge (j)
              (1,1) node (l) {} edge (h) edge (k)
              (1.5,1) node (m) {} edge (i) edge (l)
              (0,2) node (n) {} edge (k)
              (1,2) node (o) {} edge (n) edge (l)
             (0,3) node (p) {} edge (n)
              (1,3) node (q) {} edge (o) edge (p);
\draw[scale=.7,loosely dashed] (-2.5,-1) -- (2.7,-1) (-2.5,0) -- (2.7,0) (-2.5,1) -- (2.7,1) 
                      (-2.5,2) -- (2.7,2) (-2.5,3) -- (2.7,3);
  \path[scale=.7] (2.7,-1) node[coordinate,label=right:$t_1$] (t1) {}
                 (2.7,0) node[coordinate,label=right:$t_2$] (t2) {}
                 (2.7,1) node[coordinate,label=right:$t_3$] (t3) {}
                 (2.7,2) node[coordinate,label=right:$t_4$] (t4) {}
                 (2.7,3) node[coordinate,label=right:$t_5$] (t5) {};
 \end{tikzpicture}
\caption{An outerplanar graph and its track drawing.}
\label{fig:outerlay}
\end{figure}
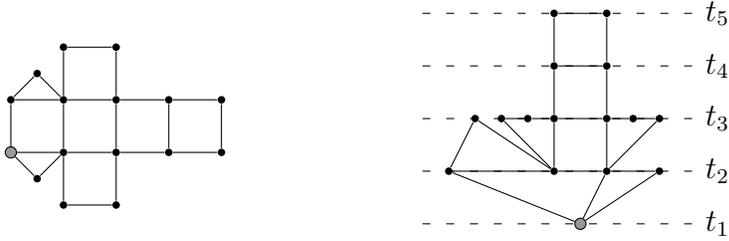

\begin{lemma}[Felsner, Liotta, and Wismath \cite{FelsnerLW03}]\label{lem:outerlayered}
Outerplanar graphs are track drawable.
\end{lemma}

\begin{lemma}\label{lem:layeredfree}
For any track drawable graph $G$ on $n$ vertices we have $\free G\ge n/2$.
\end{lemma}

\begin{proof}
Let $\pi$ be a track drawing of $G$ with tracks $t_1,\ldots,t_s$, 
lying in the plane in this order. It is practical to assume that
$t_1,\ldots,t_s$ are parallel straight line segments (rather than unbounded lines)
containing all the vertices of $\pi$. Let $\ell$ be a horizontal line.
Consider two redrawings of~$\pi$.

To make a redrawing $\pi'$, we put $t_1,t_3,t_5,\ldots$
on $\ell$ one by one. For each even index $2i$, 
we drop a perpendicular $p_{2i}$ to $\ell$ between the segments $t_{2i-1}$ and 
$t_{2i+1}$. We then put each $t_{2i}$ on $p_{2i}$ so that $t_{2i}$
is in the upper half-plane if $i$ is odd and in the lower half-plane if $i$ is
even. It is clear that such a relocation can be done so that $\pi'$
is crossing-free (the whole procedure can be thought of as sequentially unfolding 
each strip between consecutive layers to a quadrant of the plane, see the left side of
Fig.~\ref{fig:redraw}).

It is clear that the vertices on $\ell$ form a free collinear set:
if the neighboring vertices of $t_{2i-1}$ and $t_{2i+1}$ are displaced, then
$p_{2i}$ is to be shifted appropriately.

In the redrawing $\pi''$ the roles of odd and even indices are interchanged,
that is, $t_2,t_4,t_6,\ldots$ are put on $\ell$ and $t_1,t_3,t_5,\ldots$ 
on perpendiculars (see the right side of Fig.~\ref{fig:redraw}).
It remains to observe that at least one of the inequalities $\free{G,\pi'}\ge n/2$ 
and $\free{G,\pi''}\ge n/2$ must be true.
\end{proof}

\begin{figure}
\centering
 \begin{tikzpicture}[every node/.style={circle,fill,inner sep=1pt}]
  \path[scale=.6] (-1.5,0) node[circle,draw,color=black,fill=black!40,inner sep=1.5pt] (a)   {}
                 (0,3.5) node (b)   {} edge (a)
             (4,0) node (c)   {}  edge (b)
              (0,.5) node (d) {} edge (a) 
             (3,0) node (e) {} 
                 (0,2.5) node (g) {} edge (b) edge (c) edge (e)
               (0,1.5) node (h)   {} edge (a) edge (d) edge (g)
             (.5,0) node (i)   {} edge (h) 
             (2.5,0) node (j)   {} edge (e)
              (2,0) node (k)   {} edge (g) edge (j)
              (1.5,0) node (l) {} edge (h) edge (k)
              (1,0) node (m) {} edge (i) edge (l)
             (4.5,-1) node (n) {} edge (k)
             (4.5,-2) node (o) {} edge (n) edge (l)
            (5.5,0) node (p) {} edge (n)
             (6.5,0) node (q) {} edge (o) edge (p);
\draw[scale=.6,loosely dashed] (-2.5,0) -- (7.5,0) (0,.25) -- (0,4.7) (4.5,-.5) -- (4.5,-3);
  \path[scale=.6] (-1.5,0) node[coordinate,label=below:$t_1$] (t1) {}
                (0,4.7) node[coordinate,label=above:$t_2$] (t2) {}
                 (.75,0) node[coordinate,label=below:$t_3$] (t3) {}
                (4.5,-3) node[coordinate,label=below:$t_4$] (t4) {}
                 (5.75,0) node[coordinate,label=below:$t_5$] (t5) {};
 \end{tikzpicture}
\qquad\qquad\qquad
 \begin{tikzpicture}[baseline=-2.3cm,every node/.style={circle,fill,inner sep=1pt}]
  \path[scale=.6] (-1,1.5) node[circle,draw,color=black,fill=black!40,inner sep=1.5pt] (a)   {}
                 (2.5,0) node (b)   {} edge (a)
             (3,-.5) node (c)   {}  edge (b)
              (0,0) node (d) {} edge (a) 
             (3,-1) node (e) {} 
                 (1.5,0) node (g) {} edge (b) edge (c) edge (e)
               (.5,0) node (h)   {} edge (a) edge (d) edge (g)
             (3,-3.5) node (i)   {} edge (h) 
             (3,-1.5) node (j)   {} edge (e)
              (3,-2) node (k)   {} edge (g) edge (j)
              (3,-2.5) node (l) {} edge (h) edge (k)
              (3,-3) node (m) {} edge (i) edge (l)
              (4,0) node (n) {} edge (k)
              (5,0) node (o) {} edge (n) edge (l)
             (5.5,1.5) node (p) {} edge (n)
              (5.5,.5) node (q) {} edge (o) edge (p);
\draw[scale=.6,loosely dashed] (-.75,0) -- (5.35,0) 
                               (-1,.25) -- (-1,2.5) 
                               (3,-.25) -- (3,-4.5) 
                               (5.5,.25) -- (5.5,2.5);
  \path[scale=.6] (-1,2.5) node[coordinate,label=above:$t_1$] (t1) {}
                 (.25,0) node[coordinate,label=below:$t_2$] (t2) {}
                 (3,-4) node[coordinate,label=right:$t_3$] (t3) {}
                 (4.25,0) node[coordinate,label=below:$t_4$] (t4) {}
                 (5.5,2.5) node[coordinate,label=above:$t_5$] (t5) {};
 \end{tikzpicture}
\caption{Proof of Lemma \protect\ref{lem:layeredfree}:
two redrawings of the graph from Fig.~\protect\ref{fig:outerlay}.}
\label{fig:redraw}
\end{figure}
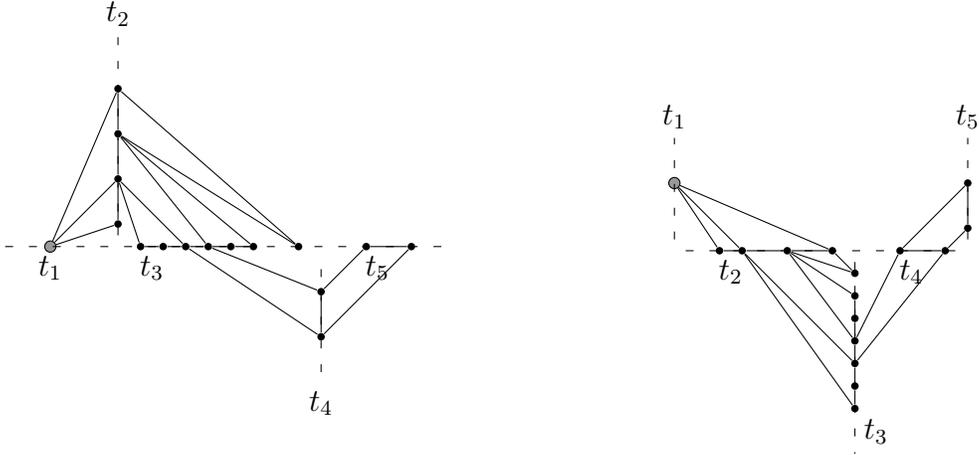

Combining Lemmas \ref{lem:outerlayered} and \ref{lem:layeredfree} with Theorem~\ref{thm:fixfree},
we obtain another proof for the following result.

\begin{corollary}[Goaoc et al.~\cite{merged}]\label{cor:outer}
For any outerplanar graph $G$ with $n$ vertices, we have $\fix G\ge \sqrt{n/2}$.
\end{corollary}

In fact, Theorem~\ref{thm:fixfree} has a much broader range of application.
The class of \emph{2-trees} is defined recursively as follows:
\begin{itemize}
\item
the graph consisting of two adjacent vertices is a 2-tree;
\item
if $G$ is a 2-tree and $H$ is obtained from $G$ by adding a new
vertex and connecting it to two adjacent vertices of $G$,
then $H$ is a 2-tree.
\end{itemize}
A graph is a \emph{partial 2-tree} if it is a subgraph of a 2-tree.
It is well known that the class of partial 2-trees coincides
with the class of graphs with treewidth at most 2.
Any outerplanar graph is a partial 2-tree, and the same holds
for series-parallel graphs (the latter class is sometimes defined
so that it coincides with the class of partial 2-trees).
Note that not all 2-trees are track drawable (for example, the graph
consisting of three triangles that share one edge).

\begin{theorem}\label{thm:2trees}
If $G$ is a partial 2-tree with $n$ vertices, then $\displaystyle\free G>n/30$.
\end{theorem}

\begin{corollary}\label{cor:2trees}
For any partial 2-tree $G$ with $n$ vertices we have $\fix G\ge \sqrt{n/30}$.
\end{corollary}

\section{The proof of Theorem \protect\ref{thm:2trees}}\label{s:2trees}

\subsection{Outline of the proof}

Given two vertex-disjoint 2-trees, we can add three edges so that the graph
obtained is a 2-tree. It readily follows that any partial 2-tree $G$ is
a spanning subgraph of some 2-tree $H$ (that is, $V_G=V_H$).
The following lemma, therefore, shows that it is enough to prove Theorem \ref{thm:2trees}
for 2-trees.

\begin{lemma}
If $G$ is a spanning subgraph of a planar graph $H$, then
$\free G\ge\free H$.
\end{lemma}

\begin{proof}
If $X$ is a free set of collinear vertices in a drawing of $H$,
then $X$ stays free in the induced drawing of~$G$.
\end{proof}

Thus, from now on we suppose that $G$ is a 2-tree. We will consider plain drawings of $G$
having a special shape. Specifically, we call a drawing \emph{folded}
if for any two triangles that share an edge, one contains the other.
Thus, all the triangles sharing an edge form a containment chain,
see Fig.~\ref{fig:folded}.
Folded drawings can be obtained by the following recursive procedure.
Suppose that $G$ is obtained from a 2-tree $G'$ by attaching a new vertex $v$
to an edge $e$ of $G'$. Then, once $G'$ is drawn, we put $v$
inside that triangular face of the current drawing of $G'$ whose boundary contains $e$.
The procedure can be implemented
in many ways giving different outputs; all of them are called folded drawings.

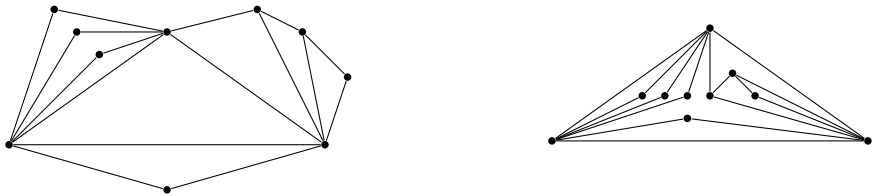
\begin{figure}
\centering
 \begin{tikzpicture}[every node/.style={circle,fill,inner sep=1pt}]
  \path[scale=.6] (0,0) node (a)   {}
                 (7,0) node (b)   {} edge (a)
             (3.5,2.5) node (c)   {} edge (a) edge (b)
                 (2,2) node (ac1) {} edge (a) edge (c)
             (1.5,2.5) node (ac2) {} edge (a) edge (c)
                 (1,3) node (ac3) {} edge (a) edge (c)
               (5.5,3) node (d)   {} edge (b) edge (c)
             (6.5,2.5) node (e)   {} edge (b) edge (d)
             (7.5,1.5) node (f)   {} edge (b) edge (e)
              (3.5,-1) node (u)   {} edge (a) edge (b);
 \end{tikzpicture}
\qquad\qquad\qquad
 \begin{tikzpicture}[baseline=-.7cm,every node/.style={circle,fill,inner sep=1pt}]
  \path[scale=.6] (0,0) node (a)   {}
                 (7,0) node (b)   {} edge (a)
             (3.5,2.5) node (c)   {} edge (a) edge (b)
                 (2,1) node (ac1) {} edge (a) edge (c)
               (2.5,1) node (ac2) {} edge (a) edge (c)
                 (3,1) node (ac3) {} edge (a) edge (c)
               (3.5,1) node (d)   {} edge (b) edge (c)
               (4,1.5) node (e)   {} edge (b) edge (d)
               (4.5,1) node (f)   {} edge (b) edge (e)
                (3,.5) node (u)   {} edge (a) edge (b);
 \end{tikzpicture}
\caption{A 2-tree and its folded drawing.}
\label{fig:folded}
\end{figure}

In fact, we will use a parallel version of the procedure producing folded drawing,
which is in Section \ref{ss:fold} called \foldn{}. The parallelization will much facilitate
proving by induction that every set of collinear vertices in a folded drawing is free.
The latter fact will reduce our task to constructing folded drawings with linear number
of collinear vertices. In Sections \ref{ss:fold1} and \ref{ss:fold2} 
we will analyze two specialized versions of \foldn{}, called
\foldn{1} and \foldn{2}, and show that, for each $G$, at least one of them succeeds
in producing a large enough collinear set.

\medskip

\paragraph{\it Some definitions.}

The vertex set and the edge set of a graph $G$ will be denoted, respectively, by
$V(G)$ and $E(G)$.
Let $G$ be a 2-tree with $n$ vertices. Thus, $G$ is obtained in $n-2$ steps
from the single-edge graph. Since exactly one new triangle appears in each step,
$G$ has $n-2$ triangles. Denote the set of all triangles in $G$ by $\triangle(G)$.
With $G$ we associate a graph $T_G$ such that $V(T_G)=E(G)\cup\triangle(G)$
and $E(T_G)$ consists of all pairs $\{uv,uvw\}$ where $uvw\in\triangle(G)$.
A simple inductive argument shows that $T_G$ is a tree.

\subsection{Folded drawings}\label{ss:fold}

A \emph{subfold} of a geometric triangle $ABC$ is a crossing-free geometric graph 
\edit{introduced?}
consisting of three triangles $A'BC$, $AB'C$, and $ABC'$ such that all of them are
inside $ABC$, see the first part of Fig.~\ref{fig:subfolds}.

A 2-tree $G$ with a designated triangle $abc\in\triangle(G)$ will be called
\emph{rooted (at the triangle $abc$)}. Since $abc$ has degree 3 in $T_G$, removal of 
$abc$ from $T_G$ splits this tree into three subtrees. Let $T_{ab}$ denote
the subtree containing the vertex $ab\in E(G)$. Note that $T_{ab}=T_{G'}$
for a certain 2-tree $G'$, a subgraph of $G$. We will denote $G'$ by $G_{ab,c}$.
The 2-trees $G_{bc,a}$ and $G_{ac,b}$ are defined symmetrically.

Now, given a 2-tree $G$, we define the class of \emph{folded drawings} of $G$.
Such a drawing is obtained by drawing an arbitrary triangle $abc\in\triangle(G)$
as a geometric triangle $ABC$ and then recursively drawing each of
$G_{ab,c}$, $G_{ac,b}$, and $G_{bc,a}$ within smaller geometric triangles
forming a subfold of $ABC$. More specifically, given a 2-tree G with root $abc\in\triangle(G)$
and a geometric triangle $ABC$, consider the following procedure.

\medskip

\shift\fold{}{G,abc,ABC}\\[-6.5mm]
\begin{itemize}
\item
draw $abc$ as $ABC$;
\item
choose three points $A'$, $B'$, and $C'$ inside $ABC$ specifying a subfold of
this triangle;
\item
designate roots $abc'$, $ab'c$, and $a'bc$ in 2-trees  $G_{ab,c}$, $G_{ac,b}$, and $G_{bc,a}$
respectively (if any of these 2-trees is empty, this branch of the procedure terminates);
\item
invoke\\
\fold{}{G_{ab,c},abc',ABC'}, 
\fold{}{G_{ac,b},ab'c,AB'C}, and  
\fold{}{G_{bc,a},a'bc,A'BC}\\
(note that the three subroutines can be executed independently in parallel).
\end{itemize}

\smallskip

It is clear that \fold{}{G,abc,ABC} produces a drawing of $G$ with outer face $ABC$.
A drawing of $G$ is called \emph{folded} if it can be obtained as the output of \fold{}{G,abc,ABC}.

Note that the described procedure is nondeterministic as we have a lot of freedom
in choosing points $A',B',C'$ and we can also have a choice of vertices $a',b',c'$.
We now introduce some notions allowing us to specify any particular computational
path of \foldn{}.

Any subfold of a geometric triangle $ABC$ specified by points $A',B',C'$
will be called a \emph{depth-1 subfold} of $ABC$, and the triangles
$A'BC$, $AB'C$, and $ABC'$ will be referred to as its triangular faces
(the outer face, which is also triangular, is not taken into account).
A \emph{depth-$(i+1)$ subfold} of $ABC$ is obtained from any depth-$i$ subfold
by subfolding each of its $3^i$ triangular faces (thereby increasing
the number of triangular faces to $3^{i+1}$), see Fig.~\ref{fig:subfolds}.

\begin{figure}
\centering
 \begin{tikzpicture}[every node/.style={circle,fill,inner sep=1pt}]
  \path[scale=.25] (0,0) node[label=left:$A$] (a)   {}
                 (16,0) node[label=right:$B$] (b)   {} edge (a)
             (8,14) node[label=above:$C$] (c)   {} edge (a) edge (b)
             (8,4.5) node[label=below:$C'$] (z)   {} edge (a) edge (b)
                 (7.5,5.5) node[label=above left:$B'$] (y) {} edge (a) edge (c)
               (8.5,5.5) node[label=above right:$A'$] (x)   {} edge (b) edge (c);
 \end{tikzpicture}
\qquad\qquad
 \begin{tikzpicture}[every node/.style={circle,fill,inner sep=1pt}]
  \path[scale=.25] (0,0) node[label=left:$A$] (a)   {}
                 (16,0) node[label=right:$B$] (b)   {} edge (a)
             (8,14) node[label=above:$C$] (c)   {} edge (a) edge (b)
             (8,4.5) node (z)   {} edge (a) edge (b)
                 (7.5,5.5) node (y) {} edge (a) edge (c)
               (8.5,5.5) node (x)   {} edge (b) edge (c)
               (5.5,6.5) node (ac)   {} edge (a) edge (c)
               (6.5,7) node (cy)   {} edge (c) edge (y)
               (5.5,5.5) node (ay)   {} edge (a) edge (y)
               (9.5,7) node (xc)   {} edge (x) edge (c)
               (10.5,6.5) node (bc)   {} edge (b) edge (c)
               (10.5,5.5) node (xb)   {} edge (x) edge (b)
               (6,2) node (az)   {} edge (a) edge (z)
               (8,1.5) node (ab)   {} edge (a) edge (b)
               (10,2) node (bz)   {} edge (b) edge (z);
 \end{tikzpicture}
\caption{Subfolds of depth 1 and 2.}
\label{fig:subfolds}
\end{figure}
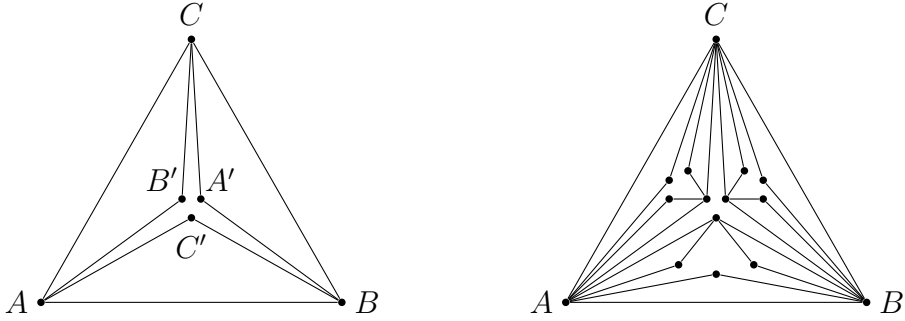

Given a 2-tree $G$, let $R$ be a tree with $V(R)=\triangle(G)$ rooted at $abc\in\triangle(G)$.
Specification of a root determines the standard parent-child relation on $\triangle(G)$.
We call $R$ a \emph{subrooting tree} for $G$ if every $xyz\in\triangle(G)$ has at most
three children and those are of the form $x'yz$, $xy'z$, $xyz'$ for some $x',y',z'\in V(G)$.

Note that removal of the root $abc$ from $R$ splits it into three subtrees
$R_{ab}$, $R_{bc}$, and $R_{ac}$ (some may be empty) such that $V(R_{ab})=\triangle(G_{ab,c})$, 
$V(R_{bc})=\triangle(G_{bc,a})$, and $V(R_{ac})=\triangle(G_{ac,b})$.
A similar fact holds true for any rooted subtree of $R$
(a rooted subtree is formed by its root $xyz$ and all descendants of $xyz$).
This observation makes evident that each computational path of \fold{}{G,abc,ABC}
is determined by a subrooting tree $R$ rooted at $abc$ and a depth-$k$ subfold $S$
of the triangle $ABC$ (where $k$ is equal to the height of $R$):
whenever the subroutine \fold{}{H,xyz,XYZ} is invoked for some 2-subtree $H$ of $G$,
we choose the vertices $x',y',z'$ and the points $X',Y',Z'$ so that
$x'yz$, $xy'z$, $xyz'$ are the children of $xyz$ in $R$, and
the geometric triangles $X'YZ$, $XY'Z$, $XYZ'$ form the subfold of $XYZ$ in $S$
(more precisely, in the fragment of $S$ of the corresponding depth).
We will denote this path of the procedure \fold{}{G,abc,ABC} by \fold{}{R,S}.

Our goal is now to establish the following fact.

\begin{lemma}\label{lem:foldfree}
Every collinear set of vertices in a folded drawing is free.
\end{lemma}

We will derive Lemma \ref{lem:foldfree} from an elementary geometric fact
stated below as Lemma \ref{lem:foldfree2},
but first we make a few useful observations and introduce some technical notions.

Note that a depth-$k$ subfold is, in a sense, unique and rigid.
More precisely, any two subfolds of the same depth are equivalent up to homeomorphisms
of the plane
(as usually, it is supposed that a homeomorphism takes vertices of geometric graphs
to vertices and edges to edges). Moreover, if $S$ is a depth-$k$ subfold of a triangle $ABC$
and a homeomorphism $h$ takes $S$ onto itself and fixes the vertices $A$, $B$, and $C$
(i.e., $h(A)=A$ etc.), then $h$ fixes every vertex of $S$.
It follows that, once we identified the vertices of the triangle by the labels
$A$, $B$, and $C$, each other vertex $X$ of $S$ can be identified by a \emph{canonical label} $c(X)$. 
Formally, a \emph{canonical labeling} of a depth-$k$ subfold $S$ is an injective map $c\function{V(S)}L$,
where a set of labels $L$ contains the labels $A$, $B$, and $C$ that are assigned to the vertices
of the outer triangle, such that any homeomorphism between two subfolds respecting the labels $A,B,C$
must respect the labels of all vertices.

A canonical labeling can be defined explicitly in many, essentially equivalent, ways.
To be specific, we can fix the following definition.
Let $S$ be a depth-$k$ subfold of a triangle $ABC$ obtained by a sequence
$ABC=S_0,S_1,\ldots,S_k=S$, where $S_{i+1}$ is obtained from $S_i$ by subfolding
each inner triangular face. Note that the sequence $S_1,\ldots,S_k$ is reconstructible
from $S$. For example, the vertex $C'\ne C$ that is adjacent to $A$ and $B$ in $S_1$
is determined by the condition that $ABC'$ is the second largest triangle in
the inclusion-chain of triangles in $S$ containing $AB$.
Now, we inductively extend the canonical labeling $c$ from $V(S_0)$ to $V(S)$
as follows:
if $XYZ$ is a triangular face of $S_{i+1}$ and $X\in V(S_{i+1})\setminus V(S_{i})$
(hence $Y,Z\in V(S_{i})$), then $c(X)=(i+1,c(YZ))$.

Furthermore, let $\ell$ be a line such that no edge of $S$ lies on it.
By $S\cap\ell$ we will denote the set of common points of $S$ and $\ell$,
consisting of the vertices of $S$ lying on $\ell$ and the points of intersection
of $\ell$ with edges of $S$. In addition to the canonical labeling of the vertices
of $S$, we label each intersection point of $\ell$ with an edge $YZ$ of $S$
by $c(YZ)$. This allows us to define the \emph{intersection pattern}
of $S$ and $\ell$ to be the sequence of the labels of all points in $S\cap\ell$
as they appear along $\ell$ (thus, this sequence is defined up to reversal).

Given two triangles $ABC$ and $A'B'C'$, we will identify the matching vertex names,
that is, $A$ and $A'$, $B$ and $B'$, and $C$ and $C'$.

\begin{lemma}\label{lem:foldfree2}
Suppose that a line $\ell$ intersects a triangle $ABC$ in at least 2 points
and the same is true for a line $\ell'$ and a triangle $A'B'C'$.
Moreover, let $ABC\cap\ell$ and $A'B'C'\cap\ell'$ have the same intersection
pattern (for example, if $\ell$ passes through $A$ and crosses $BC$, then
$\ell'$ passes through $A'$ and crosses $B'C'$).
Consider an arbitrary depth-$k$ subfold $S$ of $ABC$ and an arbitrary
set of points $P$ within $A'B'C'$ such that
$A'B'C'\cap\ell'\subset P\subset\ell'$ and $|P|=|S\cap\ell|$.
Then there exists a depth-$k$ subfold $S'$ of $A'B'C'$ such that
$S'\cap\ell'=P$ and, moreover, $S'\cap\ell'$ has the same intersection pattern
as $S\cap\ell$.
\end{lemma}

\begin{proof}
We proceed by induction on $k$. The base step of $k=1$ is a trivial
geometric graph, see Fig.~\ref{fig:basecase}.
\begin{figure}
\centering
 \begin{tikzpicture}[scale=.5,every node/.style={circle,fill,inner sep=.5pt}]
\draw (-0.5,1.5) -- (8.5,1.5);
  \path (0,0) node (a)   {}
                 (8,0) node (b)   {} edge (a)
             (4,4) node (c)   {} edge (a) edge (b)
             (4.5,1.75) node (ab)   {} edge (a) edge (b)
             (2.25,1.3) node (ac)   {} edge (a) edge (c)
             (5.5,1.7) node (bc)   {} edge (b) edge (c);
 \end{tikzpicture}
\hfill
 \begin{tikzpicture}[scale=.5,every node/.style={circle,fill,inner sep=.5pt}]
\draw (-0.5,1.5) -- (8.5,1.5);
  \path (0,0) node (a)   {}
                 (8,0) node (b)   {} edge (a)
             (4,4) node (c)   {} edge (a) edge (b)
             (2.6,2) node (ab)   {} edge (a) edge (b)
             (1.5,1.3) node (ac)   {} edge (a) edge (c)
             (4.75,1.7) node (bc)   {} edge (b) edge (c);
 \end{tikzpicture}
\hfill
 \begin{tikzpicture}[scale=.5,every node/.style={circle,fill,inner sep=.5pt}]
\draw (-0.5,1.5) -- (8.5,1.5);
  \path (0,0) node (a)   {}
                 (8,0) node (b)   {} edge (a)
             (4,4) node (c)   {} edge (a) edge (b)
             (4.4,2.9) node (ab)   {} edge (a) edge (b)
             (1.75,1.3) node (ac)   {} edge (a) edge (c)
             (6.1,1.7) node (bc)   {} edge (b) edge (c);
 \end{tikzpicture}
\caption{Base case in the proof of Lemma \protect\ref{lem:foldfree2}}
\label{fig:basecase}
\end{figure}
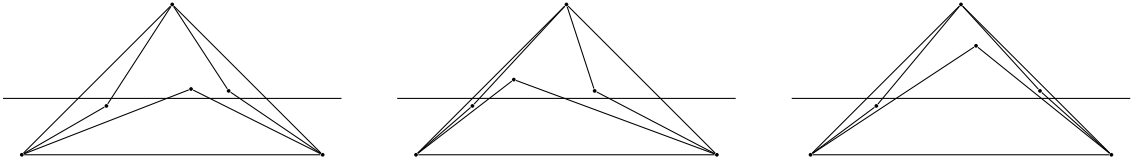
Suppose that $k>1$. Then $S$ is obtained from a depth-1 subfold $S_1$ of $ABC$
by depth-$(k-1)$ subfolding each of its three triangular faces. Label
the points of $P$ according to the intersection pattern of $S\cap\ell$.
Let $P^\circ\subseteq P$ consist of the points whose labels appear in $S_1\cap\ell$.
Since the lemma is true in the base case, there is a subfold $S'_1$ of $A'B'C'$
such that $S'_1\cap\ell'=P^\circ$ and the intersection pattern of $S'_1\cap\ell'$
agrees with the labeling of $P^\circ$.
Let $F_1,F_2,F_3$ be the triangular faces of $S_1$ and $F'_1,F'_2,F'_3$
the corresponding triangular faces of $S'_1$.
For each $i=1,2,3$, let $P_i$ be the segment of $P$ inside $F'_i$
(it may be empty for some $i$). There remains to apply the induction hypothesis
to the lines $\ell$ and $\ell'$, the triangles bounding $F_i$ and $F'_i$,
the depth-$(k-1)$ subfold of $F_i$ induced by $S$, and the point set $P_i$,
for each $i=1,2,3$.
\end{proof}

\begin{proofof}{Lemma \ref{lem:foldfree}}
Given a folded drawing $\sigma$ and a line $\ell$, we have to show that
$V(\sigma)\cap\ell$ is free. Let $ABC$ be the boundary of the outer face of $\sigma$.
It suffices to consider the case when $\sigma$ is a depth-$k$ subfold of $ABC$,
for some $k$ (completing $\sigma$, if necessary, we will prove even a stronger fact).
Using the notion of a canonical labeling, our task can be stated as follows:
Given an arbitrary set $P'$ of $|V(\sigma)\cap\ell|$ points on a line $\ell'$,
we have to find a depth-$k$ subfold $\sigma'$ of some triangle $A'B'C'$
such that $V(\sigma')\cap\ell'=P'$ and the intersection patterns of $\sigma\cap\ell$
and $\sigma'\cap\ell'$ agree on $V(\sigma)\cap\ell$ and $V(\sigma')\cap\ell'$.
The solvability of this task follows from Lemma~\ref{lem:foldfree2}.
\end{proofof}

\subsection{Folded drawings with many collinear leaf vertices}\label{ss:fold1}

Let $G$ be a 2-tree. We will call $v\in V(G)$ a \emph{leaf vertex} 
if it has degree 2. The two edges emanating from $v$ will be referred to
as \emph{leaf edges}.
A triangle containing a leaf vertex will be called a \emph{leaf triangle}. 
The number of leaf triangles in $\triangle(G)$ 
will be denoted by $t_1=t_1(G)$. 
Our nearest goal is to show that there is a folded drawing of $G$
with at least $\frac23t_1$ collinear leaf vertices.
The graph all whose triangles share an edge can be excluded from
consideration, as it can be drawn with all leaf vertices on a line.
This is the only case when a graph has only leaf triangles.
In the sequel we will, therefore, assume that $G$ contains at least one non-leaf triangle.
We will describe a specification of the procedure \foldn{}
that, additionally to a 2-tree $G$, takes on input a line $\ell$
and produces a folded drawing of $G$ such that $\ell$ crosses
each non-leaf triangle in two edges and passes through a leaf vertex
whenever possible.

\bigskip

\shift\fold1{G,\ell}\\[-7mm]
\begin{itemize}
\item
choose a non-leaf triangle $abc\in\triangle(G)$, a geometric triangle
$ABC$ with two sides crossed by $\ell$, and execute \fold{}{G,abc,ABC}
obeying the following additional conditions.
Whenever a subroutine \fold{}{H,xyz,XYZ} is invoked for a 2-subtree $H$ of $G$,
\begin{itemize}
\item
if possible,
the root $xyz$ for $H$ should be a non-leaf triangle in $G$.
In this case the geometric triangle $XYZ$ should be drawn so that two sides of it
are crossed by $\ell$;
\item
if not (i.e., $\triangle(H)$ contains only leaf triangles of $G$),
the leaf vertex of $xyz$ should be put on $\ell$ whenever this is possible.
\end{itemize}
\end{itemize}

\smallskip

\begin{lemma}\label{lem:fold1}
The procedure \fold1{G,\ell} can be run so that it produces a folded drawing
of $G$ with at least $\frac23t_1(G)$ leaf vertices lying on the line~$\ell$.
\end{lemma}

\begin{proof}
Given a non-leaf edge $e$ of $G$, let $\omega(e)$ denote the number of leaf
triangles containing $e$. For all $e$ we initially set $\cross e=1$, and
in the course of execution of \fold1{G,\ell} we reset this value to $\cross e=1$
once this edge is drawn and crosses $\ell$.
Note that, if $v$ is a leaf vertex in a leaf triangle $vuw$
and $\cross{uw}=1$, then \fold1{G,\ell} puts $v$ on $\ell$
(like any other leaf vertex adjacent to $u$ and $w$).
It follows that the procedure puts on $\ell$ at least
\begin{equation}\label{eq:sum}
\sum_e\cross e\,\omega(e)
\end{equation}
leaf vertices, where the summation goes over all non-leaf edges.
Thus, we have to show that this sum can be made large.

To this end, consider a randomized version of \fold1{G,\ell}.
Recall that the procedure begins with drawing a non-leaf triangle $abc\in\triangle(G)$ 
so that exactly two edges of it are crossed by $\ell$.
We can choose the pair of crossed edges in three ways and do it at random.
Consider now a recursive step where we have to draw a non-leaf triangle
$xyz$ whose edge $xy$ is already drawn and the vertex $z$ still not.
If $\cross{xy}=0$, then the rules of \fold1{G,\ell} force locating
$z$ so that $\cross{xz}=\cross{yz}=1$. If $\cross{xy}=1$, then exactly one
of $xz$ and $yz$ can (and must) be crossed by $\ell$. In this case
we make this choice again at random.
A simple induction (on the distance of $e$ from $abc$ in $T_G$)
shows that each non-leaf edge $e$ is crossed by $\ell$ with probability $2/3$.
By linearity of expectation, the mean value of the sum \refeq{sum} is equal to
$$
\frac23\,\sum_e\omega(e)=\frac23\,t_1(G).
$$
Therefore, the randomized version of \fold1{G,\ell} with nonzero probability puts
at least $\frac23t_1$ leaf vertices on $\ell$.
It readily follows that at least one computational path of \fold1{G,\ell}
produces a folding drawing of $G$ with at least $\frac23t_1$ leaf vertices on~$\ell$.
\end{proof}

\subsection{Folded drawings with many collinear interposed vertices}\label{ss:fold2}

Two triangles of a 2-tree $G$ will be called \emph{neighbors}
if they share an edge. Call a triangle $T\in\triangle(G)$ \emph{linking}
if it has exactly two neighbors and they are edge-disjoint.
In other words, $T$ does not share one of its edges with any other triangle
and shares each of the other two edges with exactly one triangle.
The number of linking triangles in $G$ will be denoted by $t_2=t_2(G)$.
A maximal sequence of unequal linking triangles $T_1,\ldots,T_k$ where $T_i$ and $T_{i+1}$
are neighbors for all $i<k$ will be called a \emph{chain} of triangles in $G$.
Here, \emph{maximal} means that the chain cannot be extended to a longer sequence
with the same properties. Note that different chains are disjoint.
If $k=1$, then the two neighbors of $T_1$, $T'$ and $T''$, are non-linking.
If $k>1$, then $T_1$ has exactly one non-linking neighbor $T'$, and
$T_k$ has exactly one non-linking neighbor $T''$. In either case, we say that
$T'$ and $T''$ are \emph{connected} by the chain.

\begin{lemma}\label{lem:chains}
A 2-tree $G$ with $t$ triangles has less than $t-t_2$ chains.
\end{lemma}

\begin{proof}
Define a graph $H$ on the set of all non-linking triangles of $G$ so that
two triangles are adjacent in $H$ if they are connected by a chain in $G$.
Thus, $H$ has $t-t_2$ vertices and exactly as many edges as there are chains
in $G$. It remains to notice that $H$ is acyclic (otherwise $T_G$ would contain
a cycle).
\end{proof}

The particular case of a 2-tree consisting of a single chain of length $t$
deserves a special attention. Note that, if $t\le3$, such a graph unique
up to isomorphism. Furthermore, there are two isomorphism types if $t=4$
and four isomorphism types if $t=5$. For the latter case, all possibilities
are shown in Fig.~\ref{fig:groups}, where we see also
folded drawings of the four graphs that have some useful 
properties, as stated below.

\begin{figure}
\centering
 \begin{tikzpicture}[scale=.6,every node/.style={circle,fill,inner sep=.75pt}]
  \path (0,0) node (a)   {}
                 (-2,0) node (b)   {} edge (a)
             (-1,1) node (c)   {} edge (a) edge (b)
             (1,1) node (d)   {} edge (a) edge (c)
             (2,0) node (e)   {} edge (a) edge (d)
             (1,-1) node (f)   {} edge (a) edge (e)
             (-1,-1) node (g)   {} edge (a) edge (f);
  \draw[dashed] (-1.5,.5) -- (-.5,.5) -- (d) -- (1,0) -- (.5,-.5) -- (0,-1);
  \path[shift={(-2,-4)},scale=.5] (8,0) node (a)   {}
                 (-0,0) node (b)   {} edge (a)
             (4,4) node (c)   {} edge (a) edge (b)
             (2.5,2) node (d)   {} edge (a) edge (c)
             (4,3.5) node (e)   {} edge (a) edge (d)
             (3.5,2.5) node (f)   {} edge (a) edge (e)
             (5.5,1.75) node (g)   {} edge (a) edge (f);
  \draw[shift={(-1.25,-4)},scale=.5] (-1,2) -- (6,2);
 \end{tikzpicture}
\hfill
 \begin{tikzpicture}[scale=.6,every node/.style={circle,fill,inner sep=.75pt}]
  \path (0,0) node (a)   {}
                 (-2,0) node (b)   {} edge (a)
             (-1,1) node (c)   {} edge (a) edge (b)
             (1,1) node (d)   {} edge (a) edge (c)
             (2,0) node (e)   {} edge (a) edge (d)
             (1,-1) node (f)   {} edge (a) edge (e)
             (3,-1) node (g)   {} edge (e) edge (f);
  \draw[dashed] (-1.5,.5) -- (-.5,.5) -- (d) -- (1,0) -- (1.5,-.5) -- (2,-1);
  \path[shift={(-2,-4)},scale=.5] (8,0) node (a)   {}
                 (-0,0) node (b)   {} edge (a)
             (4,4) node (c)   {} edge (a) edge (b)
             (2.5,2) node (d)   {} edge (a) edge (c)
             (4,3.5) node (e)   {} edge (a) edge (d)
             (4.75,1.5) node (f)   {} edge (a) edge (e)
             (4.83,2.29) node (g)   {} edge (e) edge (f);
  \draw[shift={(-1.25,-4)},scale=.5] (-1,2) -- (6,2);
 \end{tikzpicture}
\hfill
 \begin{tikzpicture}[scale=.6,baseline=-2.45cm,every node/.style={circle,fill,inner sep=.75pt}]
  \path (0,0) node (a)   {}
                 (-2,0) node (b)   {} edge (a)
             (-1,1) node (c)   {} edge (a) edge (b)
             (1,1) node (d)   {} edge (a) edge (c)
             (2,0) node (e)   {} edge (a) edge (d)
             (3,1) node (f)   {} edge (d) edge (e)
             (4,0) node (g)   {} edge (e) edge (f);
  \draw[dashed] (-1.5,.5) -- (-.5,.5) -- (d) -- (1,0)  (d) -- (2.5,.5) -- (3,0);
  \path[shift={(-1,-4)},scale=.5] (8,0) node (a)   {}
                 (-0,0) node (b)   {} edge (a)
             (4,4) node (c)   {} edge (a) edge (b)
             (2.5,2) node (d)   {} edge (a) edge (c)
             (4,3.5) node (e)   {} edge (a) edge (d)
             (5,1.5) node (f)   {} edge (d) edge (e)
             (3.65,2.45) node (g)   {} edge (e) edge (f);
  \draw[shift={(-.25,-4)},scale=.5] (-1,2) -- (6,2);
 \end{tikzpicture}
\hfill
 \begin{tikzpicture}[scale=.6,baseline=-2.45cm,every node/.style={circle,fill,inner sep=.5pt}]
  \path (0,0) node (a)   {}
                 (-2,0) node (b)   {} edge (a)
             (-1,1) node (c)   {} edge (a) edge (b)
             (1,1) node (d)   {} edge (a) edge (c)
             (2,0) node (e)   {} edge (a) edge (d)
             (3,1) node (f)   {} edge (d) edge (e)
             (2,2) node (g)   {} edge (d) edge (f);
  \draw[dashed] (-1.5,.5) -- (.5,.5) -- (e) -- (2,1) -- (2.5,1.5);
  \path[shift={(-1,-4)},scale=.5] (8,0) node (a)   {}
                 (-0,0) node (b)   {} edge (a)
             (4,4) node (c)   {} edge (a) edge (b)
             (2.5,2) node (d)   {} edge (a) edge (c)
             (5,2.3) node (e)   {} edge (a) edge (d)
             (6,1) node (f)   {} edge (d) edge (e)
             (3.55,1.9) node (g)   {} edge (d) edge (f);
  \draw[shift={(-1,-4)},scale=.5] (-.5,.1) -- (6.5,2.9);
 \end{tikzpicture}
\caption{Chains of length 5 and their folded drawings.
The dashed polyline shows the trace of the line $\ell$ in an unfolded drawing.}
\label{fig:groups}
\end{figure}
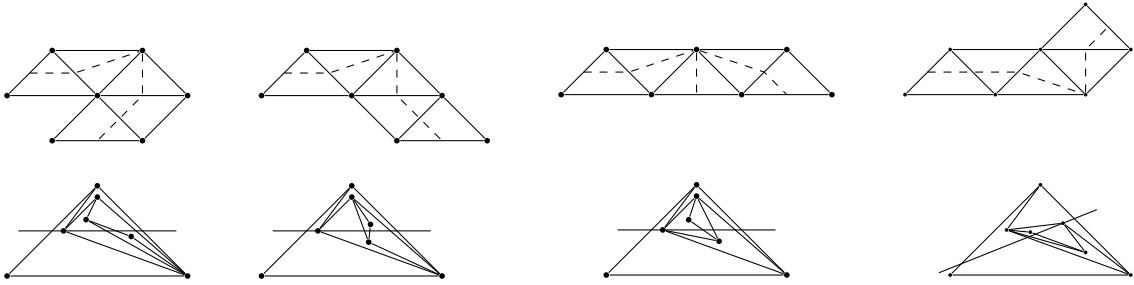

\begin{lemma}\label{lem:5chain}
Let $G$ be a 2-tree consisting of a single chain of length 5.
Given a line $\ell$, there is a folded drawing of $G$ where
$\ell$ passes through one vertex and crosses two edges in both end triangles.
\end{lemma}

Call a vertex \emph{interposed} if it belongs only to linking triangles.
We now describe a specification of the procedure \foldn{}
that aims to produce a folded drawing of a 2-tree $G$ with many 
interposed vertices on a line~$\ell$.

\medskip

\shift\fold2{G,\ell}\\[-6.5mm]
\begin{itemize}
\item
Root $G$ at a non-leaf and non-linking triangle $abc\in\triangle(G)$.
Note that, whatever subrooting tree 
\edit{recall the notion?}
is used, each chain will appear in it
as a path in the direction from the root $abc$ upwards.
Split each chain into \emph{groups} of 5 successive linking triangles,
allowing the last group to be \emph{incomplete} (i.e., have up to 4 triangles).
\item
Whenever a subroutine \fold{}{H,xyz,XYZ} is invoked for a 2-subtree $H$ of $G$,
the following rules have to be obeyed.
\begin{itemize}
\item
If possible, the root $xyz$ for $H$ should be a non-leaf triangle in $G$.
\item
If $xyz$ is neither a leaf nor a linking triangle in $G$,
then $\ell$ should cross two sides of $XYZ$. The same applies to any linking
triangle in an incomplete group.
\item
For each complete group of linking triangles, its intersection pattern
with $\ell$ should have the properties claimed by Lemma \ref{lem:5chain}.
This is always possible, irrespectively of the crossing pattern of the
triangle $x'y'z'$ preceding this group in the subrooting tree, see Fig.~\ref{fig:firstcross}
(note that $x'y'z'$ cannot be a leaf triangle and, hence, is crossed by $\ell$
in two edges).
\end{itemize}
\end{itemize}

\smallskip

\begin{figure}
\centering
 \begin{tikzpicture}[scale=.6,every node/.style={circle,fill,inner sep=.75pt}]
  \path (0,0) node (a)   {}
                 (-2,0) node (b)   {} edge (a)
             (-1,1) node (c)   {} edge (a) edge (b)
             (1,1) node (d)   {} edge (a) edge (c)
             (2,0) node (e)   {} edge (a) edge (d)
          (-1.5,.5) node[coordinate,label=above left:\small$T'$] (T0) {}
          (0,1) node[coordinate,label=above:\small$T_1$] (T1) {}
          (1.5,.5) node[coordinate,label=above right:\small$T_2$] (T2) {};
  \draw[densely dashed] (-1,0) -- (-.5,.5) -- (.5,.5);
 \end{tikzpicture}
\qquad\qquad
 \begin{tikzpicture}[scale=.6,every node/.style={circle,fill,inner sep=.75pt}]
  \path (0,0) node (a)   {}
                 (-2,0) node (b)   {} edge (a)
             (-1,1) node (c)   {} edge (a) edge (b)
             (1,1) node (d)   {} edge (a) edge (c)
             (2,0) node (e)   {} edge (a) edge (d)
          (-1.5,.5) node[coordinate,label=above left:\small$T'$] (T0) {}
          (0,1) node[coordinate,label=above:\small$T_1$] (T1) {}
          (1.5,.5) node[coordinate,label=above right:\small$T_2$] (T2) {};
  \draw[densely dashed] (-1,0) -- (-1.5,.5) (0,1) -- (.5,.5);
 \end{tikzpicture}
\caption{$T_1,T_2,\ldots,T_5$ is a group of linking triangles in a chain.
In any case, $T_1$ can be drawn so that $\ell$ crosses also the common edge of $T_1$ and~$T_2$.}
\label{fig:firstcross}
\end{figure}
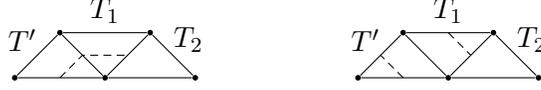

\begin{lemma}\label{lem:fold2}
Let $G$ be a 2-tree with $n$ vertices and $t_2$ linking triangles.
The procedure \fold2{G,\ell} produces a folded drawing
of $G$ with more than $t_2-\frac45n$ interposed vertices lying on the line~$\ell$.
\end{lemma}

\begin{proof}
Denote the number of interposed vertices on $\ell$ by $l$.
The rules of \foldn2 for drawing chains ensure that
$l$ is equal to the number of complete groups of linking triangles in $G$.
Let $c$ denote the number of chains in $G$, which is the trivial upper bound
for the number of incomplete groups. We, therefore, have
$$
t_2\le 5l+4c.
$$
By Lemma \ref{lem:chains}, we also have
$$
c\le n-t_2-3.
$$
It follows that $5l\ge5t_2-4n+12$, yielding the desired bound.
\end{proof}

\subsection{The rest of the proof}

Notice that the rules of \foldn1 and \foldn2 are coherent
and we can consider a hybrid procedure \fold{1+2}{G,\ell}.
This procedure aims at locating on $\ell$ as many leaf vertices
as \fold{1}{G,\ell} does and as many interposed vertices as
\fold{2}{G,\ell} does.
The latter goal is achieved by adopting the instructions of \foldn2
for drawing chains. The former goal is achieved, like \foldn1, 
by randomization. In order to ensure that, with nonzero probability,
at least $\frac23t_1(G)$ vertices are put on $\ell$,
we need to fulfill the following condition:
\begin{enumerate}
\item[(*)]
if $T$ is neither a leaf triangle nor a linking triangle in a complete group, then
each edge of $T$ is crossed by $\ell$ with probability~$\frac23$.
\end{enumerate}
The exceptional treatment of linking triangles
does not decrease the chances of any leaf vertex to be put on $\ell$.
Indeed, either a linking triangle $T$ has no leaf triangle in the neighborhood or
$T$ is the end triangle in a chain neighboring with a leaf triangle $T'$.
In the latter case the common edge of $T$ and $T'$ can be crossed, see Fig.~\ref{fig:groups}
(hence, the leaf vertex of $T'$ will be put on~$\ell$).

However, some care is needed to fulfill Condition (*)
on leaving the last complete group $T_1,T_2,T_3,T_4,T_5$ in a chain. Notice that
in each of the four cases shown in Fig.~\ref{fig:groups} we have two choices for an edge of the
end triangle $T$ that will be crossed by $\ell$. We make this choice
at random, with probability distribution shown in Fig.~\ref{fig:leaving}.
This ensures (*) for the next neighbor $T''$ of~$T_5$.

\begin{figure}
\centering
 \begin{tikzpicture}[scale=.6,every node/.style={circle,fill,inner sep=.75pt}]
  \path (0,0) node (a)   {}
                 (-2,0) node (b)   {} edge (a)
             (-1,1) node (c)   {} edge (a) edge (b)
             (1,1) node (d)   {} edge (a) edge (c)
             (2,0) node (e)   {} edge (a) edge (d)
          (-1.5,.5) node[coordinate,label=above left:\small$T_4$] (T0) {}
          (0,1) node[coordinate,label=above:\small$T_5$] (T1) {}
          (1.5,.5) node[coordinate,label=above right:\small$T''$] (T2) {};
  \path (0,2) node[coordinate,label=below:with probability $1/3$] (x)   {};
  \draw[densely dashed] (-.5,.5) -- (0,1);
 \end{tikzpicture}
\qquad\qquad
 \begin{tikzpicture}[scale=.6,every node/.style={circle,fill,inner sep=.75pt}]
  \path (0,0) node (a)   {}
                 (-2,0) node (b)   {} edge (a)
             (-1,1) node (c)   {} edge (a) edge (b)
             (1,1) node (d)   {} edge (a) edge (c)
             (2,0) node (e)   {} edge (a) edge (d)
          (-1.5,.5) node[coordinate,label=above left:\small$T_4$] (T0) {}
          (0,1) node[coordinate,label=above:\small$T_5$] (T1) {}
          (1.5,.5) node[coordinate,label=above right:\small$T''$] (T2) {};
  \path (0,2) node[coordinate,label=below:with probability $2/3$] (x)   {};
  \draw[densely dashed] (-.5,.5) -- (.5,.5);
 \end{tikzpicture}\\[-20mm]\mbox{}
\caption{On leaving the last complete group in a chain.}
\label{fig:leaving}
\end{figure}
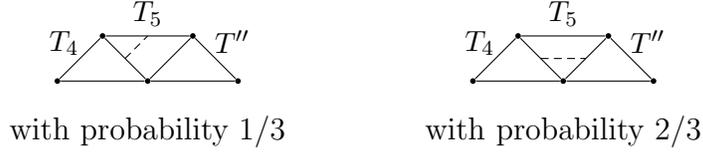

Thus, both Lemmas \ref{lem:fold1} and \ref{lem:fold2} apply as well to \foldn{1+2}
and, therefore, this procedure enables locating more than
\begin{equation}\label{eq:bound12}
\frac23\,t_1+t_2-\frac45\,n
\end{equation}
vertices of $G$ on~$\ell$.
Our further analysis is based on the following lemma.

\begin{lemma}\label{lem:t1t2}
If $G$ is a 2-tree with $n\ge3$ vertices, $t_1$ leaf triangles, and $t_2$ linking triangles,
then $4t_1+t_2>n$.
\end{lemma}

\begin{proof}
Denote the number of all triangles in $G$ by $t$ and recall that $t=n-2$.
Thus, we have to prove that
\begin{equation}\label{eq:t1t2}
4t_1+t_2\ge t+3.
\end{equation}
We proceed by induction on $t$. If $t=1$, that is, $G$ is a triangle, 
the inequality is true. Suppose that $t\ge2$.

If $G$ has a linking triangle $uvw$, where the edge $uv$ is not shared with
any other triangle, let $G'$ be obtained from $G$ by contraction of $uv$.
For the corresponding parameters of $G'$, we have $t'=t-1$, $t'_2=t_2-1$,
and $t'_1=t_1$. Inequality \refeq{t1t2} readily follows from the induction assumption.

Assume now that $t_2=0$. Consider an arbitrary leaf triangle $uvw$, with $w$ being a leaf
triangle. Our analysis is split into a few cases, see Fig.~\ref{fig:proofcases}.

\Case 1{$uvw$ has a single neighbor $uvz$.}
Since \refeq{t1t2} is true if $G$ is the diamond graph, we suppose that $n\ge5$.

\Subcase{1-a}{$uvw$ shares both edges $uz$ and $vz$ with other triangles.}
Let $G'=G-w$. We have $t'=t-1$ and $t'_1=t_1-1$. Note that the triangle $uvz$
can become linking; then we will have $t'_2=1$. The induction assumption applied to $G'$
gives us $4t_1\ge t+5$. This inequality is even stronger than \refeq{t1t2}, as $t_2=0$.

\Subcase{1-b}{$uvw$ does not share one of its edges, say $uz$, with any other triangle.}
Let $G'=G\setminus\{u,w\}$, so that $t'=t-2$. Since the triangle $uvz$ is not linking, it shares the edge
$zv$ with $k\ge2$ neighbors. This implies that $t'_1=t_1-1$. If $k=2$, the two triangles sharing $zv$ with $uvz$
can become linking in $G'$, and we will have $t'_2=2$. The induction assumption applied to $G'$
gives us $4t_1\ge t+3$, which is the same as \refeq{t1t2} because $t_2=0$.

\Case 2{$uvw$ has $k\ge2$ neighbors.} 
Let $G'=G-w$. Clearly, $t'=t-1$ and $t'_1=t_1-1$.
If $k>2$, we have $t'_2=0$. If $k=2$, it can happen that the two triangles
sharing $uv$ with $uvw$ become linking in $G'$ raising the value of $t'_2$ to 2.
Again, \refeq{t1t2} follows from the induction assumption applied to~$G'$.
\end{proof}

\begin{figure}
\centering
 \begin{tikzpicture}[scale=.6,every node/.style={circle,fill,inner sep=1pt}]
  \path (0,0) node[label=left:$z$] (z)   {}
        (2,0) node[label=right:$v$] (v)   {} edge (z)
        (1,2) node[label=above:$u$] (u)   {} edge (z) edge (v)
        (3,2) node[label=right:$w$] (w)   {} edge (u) edge (v)
        (-1,2) node (y)   {} edge (u) edge (z)
        (1,-2) node[label=below:Case 1-a] (w)   {} edge (z) edge (v);
 \end{tikzpicture}
\qquad\qquad
 \begin{tikzpicture}[scale=.6,every node/.style={circle,fill,inner sep=1pt}]
  \path (0,0) node[label=left:$z$] (z)   {}
        (2,0) node[label=right:$v$] (v)   {} edge (z)
        (1,2) node[label=above:$u$] (u)   {} edge (z) edge (v)
        (3,2) node[label=right:$w$] (w)   {} edge (u) edge (v)
        (1,-1) node (y)   {} edge (v) edge (z)
        (1,-2) node[label=below:Case 1-b] (w)   {} edge (z) edge (v);
 \end{tikzpicture}
\qquad\qquad
 \begin{tikzpicture}[baseline=-2.9cm,scale=.6,every node/.style={circle,fill,inner sep=1pt}]
  \path (0,0) node (z)   {}
        (2,0) node[label=right:$v$] (v)   {} edge (z)
        (1,2) node[label=above:$u$] (u)   {} edge (z) edge (v)
        (3,2) node[label=right:$w$] (w)   {} edge (u) edge (v)
        (.8,.6) node (y)   {} edge (v) edge (u)
        (1,-2.2) node[coordinate,label=below:Case 2] (w)   {};
 \end{tikzpicture}\\[-10mm]\mbox{}
\caption{Proof of Lemma \protect\ref{lem:t1t2} (illustrative fragments of $G$).}
\label{fig:proofcases}
\end{figure}
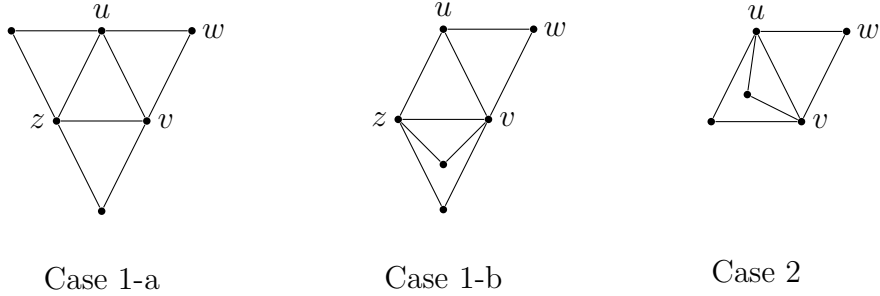

An example in Fig.~\ref{fig:4best} shows that the factor of 4 in the bound
\refeq{t1t2} cannot be improved.

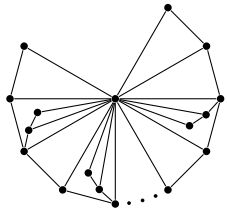
\begin{figure}[b]
\centering
 \begin{tikzpicture}[every node/.style={circle,fill,inner sep=1pt}]
  \path[scale=.7] (0,0) node (o)   {}
                (150:2cm) node (a)   {} edge (o)
                (180:2cm) node (b)   {} edge (o) edge (a)
                (-150:2cm) node (c)   {} edge (o) edge (b)
                (-160:1.75cm) node (c1)   {} edge (o) edge (c)
                (-170:1.5cm) node (c2)   {} edge (o) edge (c1)
                (-120:2cm) node (d)   {} edge (o) edge (c)
                (-90:2cm) node (e)   {} edge (o) edge (d)
                (-100:1.75cm) node (e1)   {} edge (o) edge (e)
                (-110:1.5cm) node (e2)   {} edge (o) edge (e1)
                (-60:2cm) node (f)   {} edge (o)
                (-82.5:2cm) node[inner sep=.5pt] (f1)   {}
                (-75:2cm) node[inner sep=.5pt] (f2)   {}
                (-67.5:2cm) node[inner sep=.5pt] (f3)   {}
                (-30:2cm) node (g)   {} edge (o) edge (f)
                (0:2cm) node (h)   {} edge (o) edge (g)
                (-10:1.75cm) node (h1)   {} edge (o) edge (h)
                (-20:1.5cm) node (h2)   {} edge (o) edge (h1)
                (30:2cm) node (i)   {} edge (o) edge (h)
                (60:2cm) node (j)   {} edge (o) edge (i);
 \end{tikzpicture}
\caption{A 2-tree with  $t=4k+2$ triangles, for which $t_1=k+2$ and $t_2=0$
(the parameter $k$ can take an arbitrary value).}
\label{fig:4best}
\end{figure}

Turning back to the proof of Theorem \ref{thm:2trees}, suppose first that $t_2\le\frac45n$.
In this case the bound \refeq{bound12} has no advantage upon the performance of
\fold1{G,\ell}, that puts at least $\frac23t_1$ vertices on $\ell$.
By Lemma \ref{lem:t1t2}, we then have $t_1>\frac{n-t_2}4\ge\frac n{20}$,
and hence $\ell$ passes trough more than $n/30$ vertices.
If $t_2>\frac45n$, the procedure \foldn{1+2} is preferable and yields
$$
\frac23\,t_1+t_2-\frac45\,n>\frac{n-t_2}6+t_2-\frac45\,n=\frac{25t_2-19n}{30}>\frac n{30}.
$$
collinear vertices. The proof is complete.

\section{Questions and comments}\label{s:open}
\mbox{}

\que
How far or close are parameters $\free G$ and $\lin G$?
It seems that a priori we even cannot exclude equality.
To clarify this question, it would be helpful to (dis)prove that every collinear set 
in any straight line drawing is free.

\que
We constructed examples of graphs with $\free G\le\lin G\le O(n^{\sigma+\epsilon})$
for a graph-theoretic constant $\sigma$, for which it is known that
$0.753<\sigma<0.99$. Are there graphs with $\lin G=O(\sqrt n)$?
If so, this could be considered a strengthening of the examples of graphs
with $\fix G=O(\sqrt n)$ given in \cite{Bose,merged,KPRSV}.
Are there graphs with, at least, $\free G=O(\sqrt n)$? If not,
by Theorem \ref{thm:fixfree} this would lead to an improvement of 
Bose et al.'s bound~\refeq{bose}.

\que
By Theorem \ref{thm:2trees}, we have $\free G\ge n/30$
for any graph $G$ with tree-width no more than 2.
One can also show that for Halin graphs, whose tree-width can attain 3,
we have $\free G\ge n/2$.
For which other classes of graphs do we have $\free G=\Omega(n)$
or, at least, $\lin G=\Omega(n)$? In particular, is $\lin G$ linear
for 2-outerplanar graphs? These graphs have tree-width at most 5,
and one can extend this question to planar graphs with tree-width bounded by a small
constant $t$. Corollary \ref{cor:btw} gives a negative answer if $t$ is sufficiently large.
Furthermore, what about planar graphs with bounded vertex degrees?
Note that the graphs constructed in the proof of Theorem \ref{thm:lin}
have vertices with degree more than $n^\delta$ for some $\delta>0$.

\que
In a recent paper \cite{CanoTU11},
Cano, T\'oth, and Urrutia improve the upper bound \refeq{fixupper}
to a bound of $O(n^{1/(3-\sigma)+\epsilon})$. Similarly to our proof of Theorem \ref{thm:lin},
their construction also uses iterative refinement
of faces of a planar triangulation.
It follows that our lower bound of $O(\sqrt n)$ in Corollary \ref{cor:2trees} 
cannot be extended to any class of planar graphs with bounded tree-width,
even to planar graphs of tree-width 8. If such extension is possible
for tree-width 3,4,\ldots is a natural open problem.

\que
Whether or not $\fit G=\lin G$ is an intriguing question.
Similarly to \refeq{fixfixl}, one can prove that $\fit G=\fitl G$,
where $\fitl G=\min_{X}\fitx G$ with the minimization over collinear $X$.
Thus, the question is actually whether or not
$\fitx G$ has the same value for all collinear $X$.
A similar question for $\fixx G$ is also open; it is posed in \cite[Problem 6.5]{KPRSV}.
\edit{
These questions have a strong flavour of \emph{morphing} issues, see, e.g., \cite{}.
}

\que
It is also natural to consider $\fitg G=\min_X\fitx G$, where the minimization goes over
all $X$ in general position. Can one extend our upper bound $\fit G=O(n^{0.99})$ to show
that $\fitg G=o(n)$ for infinitely many $G$?

\hide{
\que
Since grid graphs are almost layered, we have for them $\fix{P_s\times P_s}\ge\sqrt{n/2}$
where $n=s^2$.
How tight is this lower bound?
From \cite[Corollary 4.1]{Cib} we know that $\fix{P_s\times P_s}=O((n\log n)^{2/3})$.
}

\que
By slightly modifying the proof of Lemma \ref{lem:Gk}.2,
one can show that $c(G_k^*)\le(c(G_1^*)-1)^{k-1}c(G_1^*)$.
It follows that, for any cubic polyhedral graph $H$,
there exists a sequence of cubic polyhedral graph $H_1,H_2,\ldots$
such that $\lim_{k\to\infty}\frac{\log c(H_k)}{\log v(H_k)}\le
\frac{\log (c(H)-1)}{\log (v(H)-1)}<\frac{\log c(H)}{\log v(H)}$.
This readily implies two properties of the shortness exponent
for cubic polyhedral graphs, that seem to be unnoticed so far.
First, $\sigma=\inf_H\frac{\log c(H)}{\log v(H)}$ over all
cubic polyhedral $H$. Second, $\sigma$ cannot be attained
by the fraction $\frac{\log c(H)}{\log v(H)}$ for any particular~$H$.
It is interesting if this holds true for other families of graphs.


\begin{thebibliography}{10}


\bibitem{BilinskiJMY11}
M.~Bilinski, B.~Jackson, J.~Ma, X.~Yu.
\newblock
Circumference of 3-connected claw-free graphs and large Eulerian subgraphs 
of 3-edge-connected graphs.
\newblock
{\em Journal of Combinatorial Theory, Series B} 101(4):214--236 (2011).


\bibitem{Bodlaender98}
H.L.~Bodlaender.
\newblock
A partial $k$-arboretum of graphs with bounded treewidth.
\newblock
{\it Theoretical Computer Science} 209(1-2):1--45 (1998).


\bibitem{Bose}
P.~Bose, V.~Dujmovic, F.~Hurtado, S.~Langerman, P.~Morin, D.R.~Wood.
\newblock
A polynomial bound for untangling geometric planar graphs. 
\newblock
{\em Discrete and Computational Geometry} 42(4):570--585 (2009).

\bibitem{CanoTU11}
J.~Cano, C.D.~T\'oth, J.~Urrutia.
\newblock
Upper bound constructions for untangling planar geometric graphs.
\newblock
Proc.\ of the 19-th Int.\ Symp.\ on {\em Graph Drawing, 2011},
to appear.

\bibitem{Cib}
J.~Cibulka.
\newblock
Untangling polygons and graphs.
\newblock
{\em Discrete and Computational Geometry} 43(2):402--411 (2010).

\bibitem{FelsnerLW03}
S.~Felsner, G.~Liotta, S.K.~Wismath.
\newblock
Straight-line drawings on restricted integer grids in two and three dimensions. 
\newblock
{\em J.~Graph Algorithms Appl.} 7(4):363--398 (2003).


\bibitem{FlajoletN99}
P.~Flajolet, M.~Noy.
\newblock Analytic combinatorics of non-crossing configurations.
\newblock {\em Discrete Mathematics}, 204(1-3):203--229, 1999.

\bibitem{GarciaHHTV09}
A.~Garc\'{\i}a, F.~Hurtado, C.~Huemer, J.~Tejel, P.~Valtr.
\newblock On triconnected and cubic plane graphs on given point sets.
\newblock {\em Comput. Geom.} 42(9):913--922, 2009.

\bibitem{GimenezN09}
O.~Gim\'enez, M.~Noy.
\newblock {Counting planar graphs and related families of graphs.}
\newblock In: {\it Surveys in combinatorics 2009.}
Cambridge: Cambridge University Press, pages 169--210 (2009).

\bibitem{merged}
X.~Goaoc, J.~Kratochv\'{\i}l, Y.~Okamoto, C.S.~Shin, A.~Spillner, A.~Wolff.
\newblock
Untangling a planar graph.
\newblock
{\it Discrete and Computational Geometry\/} 42(4):542--569 (2009).

\bibitem{GritzmannMPP91}
P.~Gritzmann, B.~Mohar, J.~Pach, R.~Pollack.
\newblock Embedding a planar triangulation with vertices at specified points.
\newblock {\em American Mathematical Monthly} 98(2):165--166, 1991.

\bibitem{Gru}
B.~Gr\"unbaum.
\newblock
How to cut all edges of a polytope?
\newblock
{\it The American Mathematical Monthly} 79(8):890--895 (1972).

\bibitem{GWa}
B.~Gr\"unbaum, H.~Walther.
\newblock
Shortness exponents of families of graphs.
\newblock
{\it J.\ Combin.\ Theory A} 14:364--385 (1973).

\bibitem{KPRSV}
M.~Kang, O.~Pikhurko, A.~Ravsky, M.~Schacht, O.~Verbitsky.
\newblock
Untangling planar graphs from a specified vertex position
--- Hard cases.
\newblock
{\it Discrete Applied Mathematics} 159:789--799 (2011).


\bibitem{PTa}
J.~Pach, G.~Tardos.
\newblock
Untangling a polygon. 
\newblock
{\it Discrete and Computational Geometry\/} 28:585--592 (2002).

\bibitem{Ver}
O.~Verbitsky.
\newblock
On the obfuscation complexity of planar graphs.
\newblock
{\it Theoretical Computer Science\/} 396(1--3):294--300 (2008).

\end{thebibliography}
\end{document}